\documentclass[a4paper,10pt,twoside,reqno,nonamelimits]{article}
\usepackage{a4wide}

\usepackage[T1]{fontenc}
\usepackage[utf8]{inputenc}
\usepackage{newcent,dsfont,bbold} 

\usepackage{authblk}

\usepackage[estonian,english]{babel} 

\usepackage[fleqn]{amsmath}
\usepackage{amssymb,amsfonts,newlfont,amscd,amsthm,amsxtra,mathrsfs,nicefrac}

\usepackage{graphicx,xcolor}
\usepackage{booktabs,float,subfigure}
\usepackage{MnSymbol}

\usepackage[norelsize,boxed,noend,linesnumbered]{algorithm2e}\DontPrintSemicolon
\SetKwFor{RepeatTimes}{repeat}{times}{endrepeat}
\SetKw{given}{Given}
\SetKw{Initialize}{Initialize}
\RestyleAlgo{ruled}

\usepackage[pagebackref,colorlinks=true,urlcolor=blue,linkcolor=blue,citecolor=blue]{hyperref}
\usepackage{cite,url}

\usepackage[shortlabels]{enumitem} % enumiration package
\DeclareMathOperator*{\Exp}{\mathbb{E}}
\DeclareMathOperator*{\Var}{\mathbb{Var}}
\newcommand{\Min}[1]{\min\limits_}
\newcommand{\abs}[1]{{\lt\lvert{#1}\rt\rvert}}
\DeclareMathOperator{\Inf}{Inf}
\DeclareMathOperator{\Stab}{stab}
\DeclareMathOperator{\Sens}{sens}
\DeclareMathOperator{\NS}{NS}
\DeclareMathOperator{\MI}{MI}
\DeclareMathOperator{\CoI}{CoI}

\DeclareMathOperator{\CI}{CI}
\DeclareMathOperator{\UI}{UI}
\DeclareMathOperator{\SI}{SI}

\DeclareSymbolFont{symbolsC}{U}{txsyc}{m}{n}
\DeclareMathSymbol{\notni}{\mathrel}{symbolsC}{61}

\newcommand{\Expl}[1]{\Exp\limits_}

\newcommand{\RV}[1]{\mathbf{#1}}

\newcommand{\RR}{\mathbb{R}}
\newcommand{\lt}{\left}
\newcommand{\rt}{\right}
\newcommand{\bMtx}[1]{{\begin{pmatrix}#1\end{pmatrix}}}
\newcommand{\nfrac}[2]{{\nicefrac{#1}{#2}}}
\theoremstyle{plain}
\newtheorem{theorem}{Theorem}[section]
\newtheorem{proposition}{Proposition}[section]
\newtheorem{corollary}{Corollary}[section]

\newtheorem{conjecture}{Conjecture}[section]

\theoremstyle{definition}

\newtheorem{remark}{Remark}[section]
\newtheorem*{remark*}{Remark}

\newtheorem*{example*}{Example}

\title{Partial Information Decomposition of Boolean Functions: a Fourier Analysis perspective}
\author[1]{Abdullah Makkeh}
\author[2]{Dirk Oliver Theis}
\author[2]{Raul Vicente}
\affil[1]{\small Campus Institute for Dynamics of Biological Networks, Georg-August University, G\"{o}ttingen, Germany}
\affil[2]{\small Institute of Computer Science {\tiny of the} University of Tartu, Tartu, Estonia}
\date{}
\begin{document}
	% == hyperref blue ==
	\hypersetup{linkcolor=blue}
		
	% == Declare Title ==
	\maketitle
	
	% == Abstract ==
	\begin{abstract}
	    Partial information decomposition (PID) partitions the information that a set of sources has about a target variable into synergistic, unique, and redundant contributions. This information-theoretic tool has recently attracted attention due to its potential to characterize the information processing in multivariate systems. However, the PID framework still lacks a solid and intuitive interpretation of its information components. In the aim to improve the understanding of PID components, we focus here on Boolean gates, a much studied type of source--target mechanisms. Boolean gates have been extensively characterised via Fourier analysis which coefficients have been related to interesting properties of the functions defining the gates. In this paper we establish for Boolean gates mechanisms a relation between their PID components and Fourier coefficients.
	\end{abstract}
	% == Body text ==
	% % % % % % %
	% % % % % % %
	% 
\section{Motivation}
Boolean functions map each sequence of bits to a single "0" or "1" \cite{crama2011boolean}. Thus, this rich family of functions can encode any property that either holds or not for each possible binary string. Their study is a major topic in mathematics and theoretical computer science \cite{jukna2012boolean,o2014analysis}, and they are often used to model physical, social, and biological, collective and network phenomena \cite{kluver1999topology,kauffman1993origins,wang2012boolean}. A major tool to analyse Boolean functions is via their spectral or Fourier properties \cite{o2014analysis}. The Fourier decomposition of a Boolean function corresponds to its expansion as a multilinear polynomial on the basis of all possible products of distinct variables. Different coefficients of such an expansion measure the correlation of the function with different combinations of input variables \cite{jukna2012boolean}. These Fourier coefficients have been directly related to many important properties of the function, including its average sensitivity to the flipping of the input bits, also known as influence \cite{o2014analysis}.

On a different perspective, the partial information decomposition (PID) aims to describe how information about one random variable is distributed among a set of other random variables \cite{Williams10,Bertschinger12}. In particular, the mutual information between one target random variable $(\RV{T})$ and a pair of source random variables $(\{\RV{X_1}, \RV{X_2}\})$ can be decomposed in the following components: the unique information that each source holds about the target, the redundant information that the sources share about the target, and the synergistic information that the combination of sources produce about the target \cite{Bertschinger12}. The PID framework has been used to study the interactions that emerge in the analysis of several complex systems such as gene networks~\cite{anastassiou2007computational,chatterjee2016construction, watkinson2009inference}, interactive agents~\cite{ay2012information,flack2012multiple,frey2018synergistic,katz2011inferring}, or neural processing~\cite{faes2016information,marre2009prediction,pica2017quantifying,wibral2017partial}. In a wider scope, the complexion of the information hold by the inputs identify the complexity of extracting it~\cite{schneidman2003synergy,ver2017disentangled}, its robustness to disruptions of the system~\cite{rauh2014robustness}, or how to reduce the inputs dimensionality without information loss~\cite{banerjee2018variational,tishby2000information}.

The PID framework can be applied to the output and input variables of any Boolean function or logic gate once these variables are equipped with a probability distribution. Importantly, recent developments have extended the partial information decomposition to the multivariate case~\cite{chicharro2017quantifying,Ince16,makkeh2019maxent3d_pid,finn2018pointwise, makkeh2020differentiable}. Both decompositions (Fourier and PID) have been previously applied to characterize the information processing by logic gates from different perspectives, and hence a natural question is whether both decompositions are anyhow related.

The main aim of this work is to describe the relation between the partial information decomposition and the Fourier analysis of Boolean functions. Such mapping provides one way to translate and interpret results obtained in the Fourier analysis of Boolean functions to the information theoretic lens of PID. Given the ongoing discussions on several versions of PIDs and the interpretation of its terms, one broader objective of this work is to spark an interest in studying PIDs via their relation to other decompositions or representations which are better understood.

The paper is organized as follows. In the Background section (Section~\ref{sec:back}) we describe the basic notions of Fourier decomposition of Boolean functions and the general framework of PID. In Section~\ref{sec:the-mapping}, we show the mapping between the Fourier coefficients and the different PID terms and present the main technical results about the mapping for bivariate and trivariate functions in the case of uniform measures. Section~\ref{sec:apps} discusses some applications of such a mapping, while Section~\ref{sec:p-biased} generalizes the results for the case for p-biased measures of the input variables. Finally, some conclusions and future directions are discussed.

\section{Background}\label{sec:back}
This section is an overview of notions and tools necessary for the contributions presented in the paper. We start with a brief introduction to Fourier analysis of Boolean functions and state the main results that are needed later. Then, we present the relation between Fourier analysis and mutual information of Boolean functions that was studied by Heckel et al.~\cite{heckel2013harmonic}. Finally, we explain the concept of partial information decomposition, \emph{PID}, which aims to disentangle the information contributions of a set sources about a target variable into non-negative unique, redundant, and synergistic components of information. In particular, we review the PID framework introduced by Williams and Beer~\cite{Williams10}.

\subsection{Fourier Analysis of Boolean Functions}\label{subsec:bg-fou}
A Boolean function $f$ has a unique Fourier transform as a multilinear polynomial. To make the notation of the Fourier transform easier the domain of the Boolean functions is taken in $\{-1,1\}$, i.e.~, $f:\{-1,1\}^n\to\{-1,1\}$. \\
		
\noindent The Fourier expansion of $f$ is given as 
		\begin{equation}
			f(\RV{X}) = \sum_{S\subseteq[n]}\hat{f}(S)\Phi_S(\RV{X}),\quad 
			\Phi_S(\RV{X}) :=	\begin{cases}
									\prod_{i\in S}\RV{X_i} 	&~\text{if $S\neq\emptyset$}\\
									1						&~\text{if $S=\emptyset$}.
								\end{cases}
		\end{equation}

\noindent where $[n]$ is the standard shorthand notation for the set $\{0,\dots,n-1\}$ and $\hat{f}$ are the Fourier coefficients. Parseval's identity implies that the Fourier coefficients $\hat{f}$ define a probability distribution over $S\subseteq [n],$ namely, 
		\begin{equation}\label{eq:parseval}
			\Expl_\RV{X}[f(\RV{X})^2]= \sum_{S\subseteq[n]} \hat{f}(S)^2 = 1.
		\end{equation}
		
Using equation~\eqref{eq:parseval}, the variance of $f$ can be written in terms of its Fourier coefficients
		
		\begin{equation}
			\Var[f] = \Expl_\RV{X}[f(\RV{X})^2] - \Expl_\RV{X}[f(\RV{X})]^2 = \sum_{S\subseteq[n]} \hat{f}(S)^2 - \hat{f}(\emptyset)^2 = \sum_{S\subseteq[n]\backslash\emptyset} \hat{f}(S)^2.
		\end{equation}
		
The influence of a source $\RV{X_i}$ over the value of the target $f(x_1,\dots,x_n)$, $\Inf_i[f]$, is defined as the probability that flipping the $i$th source flips the value of the function (target). This influence $\Inf_i[f]$ can be determined using the Fourier transform~\cite[Theorem 2.20]{o2014analysis} 

		\begin{equation}
			\Inf_i[f]= \sum_{S\subseteq[n]: i\in S} \hat{f}(S)^2.
		\end{equation}
		
\noindent Then, the influence of a group $A$ of sources is
		\begin{equation}
			\Inf_A[f]= \sum_{S\subseteq [n]}\abs{S\cap A}\hat{f}(S)^2.
		\end{equation}
		Setting $A =[n],$ $\Inf_A[f]$ is called the total influence and denoted by $\Inf[f]$ where 
		\begin{equation}
			\Inf[f] = \sum_{S\subseteq[n]}\abs{S}\hat{f}(S)^2.
		\end{equation}
		
\noindent This total influence can be seen as the average sensitivity of $f,$ formally described in the following proposition.
		
		\begin{proposition}[Proposition 2.28~\cite{o2014analysis}]\label{prop:tot-inf}
			For $f: \{-1,1\}^n\to\{-1,1\}$
			
			\begin{equation*}
				\Inf[f]= \Expl_\RV{X}[\Sens_f(\RV{X})],
			\end{equation*}
			
        \noindent where $\Sens_f(\RV{X})$ is the sensitivity of $f$ at $\RV{X}$, defined to be the number of pivotal coordinates for $f$ on input $\RV{X}$, i.e., $\Sens_f = \displaystyle\sum_{i=1}^n \mathbf{1}_{f(\RV{X})\neq f(\RV{X}^{\oplus i})}$	
		\end{proposition}
		
Lower bounds for the the total influence known as Poincar\'e Inequality can be obtained using the variance of any $f : \{-1,1\}^n \to\RR$. 		
		\begin{proposition}[Poincar\'e Inequality]\label{prop:var-inf}
			For any $f : \{-1,1\}^n \to\RR$, $\Var[f]\leq\Inf[f]$.
		\end{proposition}		
A sharper lower bound can be obtained for Boolean functions as expressed in following theorem.		
		\begin{theorem}[Theorem 2.39~\cite{o2014analysis}]\label{thm:tot-inf}
			For $f : \{-1,1\}^n \to\{-1,1\}$ with $\alpha = \min\{P[f=1],P[f=-1]\}$,
			\begin{equation}
				2\alpha\log(\frac{1}{\alpha})\leq \Inf[f].
			\end{equation}
		\end{theorem}
		
Since influence of sources will play a key role in the rest of the paper, next we state the influences $\Inf_i$ for some important families of Boolean functions. 

		\begin{proposition}[Proposition 2.21~\cite{o2014analysis}]\label{prop:mono-inf}
			Let $f : \{-1,1\}^n \to\{-1,1\}$ be a monotone function. Then, 
			\begin{equation*}
			    \Inf_i[f] = \begin{cases}
			                  \hat{f}(\{i\})    &~\text{if $f$ is increasing}\\
			                  -\hat{f}(\{i\})    &~\text{if $f$ is decreasing.}\\
			                \end{cases}
			\end{equation*}
		\end{proposition}

		\begin{proposition}[Proposition 2.22~\cite{o2014analysis}]\label{prop:tran-sym-mono-inf}
			Let $f : \{-1,1\}^n \to\{-1,1\}$ be a transitive symmetric and monotone function. Then, $\Inf_i[f]\le \frac{1}{\sqrt{n}}$ for all $i\in[n]$.
		\end{proposition}
		
        A Boolean function $f:\{-1,1\}^n\to\{-1,1\}$ is said to be unate in direction $i$ if for each $(x_1,\dots,x_n)$ and a fixed $a_i\in\{-1,1\}$, $f(x_1,\dots,x_i=-a_i,\dots,x_n)\le f(x_1,\dots,x_i=a_i,\dots,x_n)$ holds. The function $f$ is said to unate if $f$ is unate in direction $i$ for each $i\in[n].$
        
        \begin{proposition}[Proposition 3~\cite{heckel2013harmonic}]\label{prop:unate}
            Let $f:\{-1,1\}^n\to\{-1,1\}$ be unate. Then,
            $$\hat{f}(\{x_i\}) = a_i\Inf_i[f],$$
            where $a_i\in\{-1,1\}$ is the unate parameter. 
        \end{proposition}
        
        \begin{remark}\label{rem:unate}
            Every linear threshold function is unate.
        \end{remark}
        
We will conclude this subsection by presenting the notion of stability in Boolean functions. Let $ \rho\in[0,1]$, for fixed $\RV{X} \in\{-1,1\}^n$ we write $\RV{Y} \sim N_\rho(\RV{X})$ to denote that the random string $\RV{Y}$ is drawn as follows: 
		for each $i\in[n]$ independently,
		\begin{equation*}
			\RV{Y}_i =	\begin{cases}
						\RV{X}_i 				&~\text{with probability}~\rho\\
						\text{uniformly random}	&~\text{with probability}~1-\rho.
				 	\end{cases}
		\end{equation*}
		
		The notation is extended to all $\rho\in[-1,1]$ as follows:
		\begin{equation*}
			\RV{Y}_i =	\begin{cases}
							\RV{X}_i	&~\text{with probability}~\frac{1}{2} +\frac{1}{2}\rho\\
							-\RV{X}_i	&~\text{with probability}~\frac{1}{2} -\frac{1}{2}\rho.
						\end{cases}
		\end{equation*}
		
Then, $\RV{Y}$ and $\RV{X}$ are said to be $\rho$-correlated. The noise stability of $f$ at $\rho$ $(\Stab_\rho)$ measures the correlation between $f(\RV{X})$ and $f(\RV{Y})$ such that $\RV{X}$ and $\RV{Y}$ are $\rho$-correlated. The Fourier coefficients of $f$ are related to the stability of $f$ at $\rho$~\cite[Theorem 2.49]{o2014analysis} by 
		\begin{equation}\label{eq:stab}
			\Stab_\rho[f]=\sum_{S\subseteq[n]} \rho^{\abs{S}}\hat{f}(S)^2.  
		\end{equation}
For $f : \{-1,1\}^n \to\{-1,1\}$ and $\delta\in[0,1]$, $\NS_\delta[ f ]$ denotes the noise sensitivity of $f$ at $\delta$, defined to be the probability that $f(\RV{X})\neq f(\RV{Y})$ when $\RV{X}\in\{-1,1\}^n$ is uniformly random and $\RV{Y}$ is formed from $\RV{X}$ by reversing each bit independently with probability $\delta$. Noise sensitivity and stability are directly related by the expression:
		\begin{equation}
			\NS_\delta[f] = \frac{1}{2} - \frac{1}{2}\Stab_{1-2\delta}[f].
		\end{equation}
	% % % % % % % %
	% % % % % % % %

\subsection{Information and Influence}\label{subsec:heckel}
	Consider $(\RV{T},\RV{X}_1,\dots,\RV{X}_n)$ to be a multivariate system (gate) and $P$ be the joint probability of $(\RV{T},\RV{X}_1,\dots,\RV{X}_n)$. This system is said to be Boolean if for any $x_1,\dots,x_n\in X_1\times \dots\times X_n$ there exists one and only one $t\in T$ such that $P(t,x_1,\dots,x_n)>0$ where the alphabets $X_i$ and $T$ are binary sets. 
	    
Let $f:\{-1,1\}^n\to\{1,1\}$ be a Boolean function representing the multivariate Boolean system $(\RV{T},\RV{X}_1,\dots,\RV{X}_n)$, where $\RV{T}=f(\RV{X}_1,\dots,\RV{X}_n)$.

Heckel et al.~\cite[Theorem 1]{heckel2013harmonic} showed that $H(f(\RV{X}_1,\dots,\RV{X}_n)\mid \RV{A})$, i.e. the conditional entropy of the target conditioned on a group $A$ of sources where $\RV{A} = \{\RV{X}_i\mid i\in A\subseteq[n]\}$, is a function of the Fourier coefficients of $f$ 
	    
		\begin{equation}
			H(f(\RV{X}_1,\dots,\RV{X}_n)\mid \RV{A}) = \Exp\left[h\left(\frac{1}{2}\left(1 + \sum_{S\subseteq A}\hat{f}(S)\Phi_S(\RV{A})\right)\right)\right],
		\end{equation}
		
\noindent where $h(x) = x\log x + (1-x)\log(1-x)$. In~\cite[Corollary 1]{heckel2013harmonic}, the authors of the study expressed the mutual information $\MI(f(\RV{X_1},\dots,\RV{X_n});\RV{A})$
in terms of the Fourier coefficients of $f$,
		\begin{equation}\label{eq:exp-mi}
			\MI(f(\RV{X}_1,\dots,\RV{X}_n);\RV{A}) = h(\frac{1}{2}(1 + \hat{f}(\emptyset)))-\Exp\left[h\left(\frac{1}{2}\left(1 + \sum_{S\subseteq A}\hat{f}(S)\Phi_S(\RV{A})\right)\right)\right].
		\end{equation}
		
Using~\cite[Theorem 1]{heckel2013harmonic}, it is possible to deduce the following relation between the influence of a source and some entropy quantities~\cite[Theorem 4]{heckel2013harmonic}
		\begin{equation}\label{eq:inf-mi}
			\Inf_i[f] = \frac{H(f(\RV{X}_1,\dots,\RV{X}_n)\mid \RV{A}_i)}{H(\RV{X}_i)},
		\end{equation}
		
\noindent where $\RV{A}_i= (\RV{X}_1,\dots,\RV{X}_{i-1},\RV{X}_{i+1}, \dots, \RV{X}_n)$. Note that when the sources of the Boolean functions are uniformly distributed, then it is possible to express influence as the following conditional mutual information: 
		\begin{equation}\label{eq:inf-mi-uni}
			\Inf_i[f] = \MI(f(\RV{X}_1,\dots,\RV{X}_n);\RV{X}_i\mid \RV{A}_i).
		\end{equation}
    % % % % % % % %
    % % % % % % % %
\subsection{Partial Information Decomposition}
Partial information decomposition (PID) is a recent information-theoretic framework developed to capture different types of information dependencies within multivariate systems. In particular, it can be applied to quantify the information contributions of a set of sources $(\RV{X}_1,\dots,\RV{X}_n)$ about the target $\RV{T}$ in a system $(\RV{T},\RV{X}_1,\dots,\RV{X}_n).$ In essence, PID quantifies the dependencies between the sources about the target into redundant, unique, and synergistic information.
        
Williams and Beer~\cite{Williams10} introduced the PID framework which defines a decomposition via several axioms encoding a set of desired properties for any meaningful decomposition.  However, this framework alone is not enough to compute the individual terms of the decomposition. Thus, Williams and Beer also proposed the first measure to compute the individual terms of the PID, and since then different measures have followed~\cite{Harder12,Williams10,Bertschinger12,finn2018pointwise,Chicharro17,Griffith13,Ince16,James17,makkeh2020differentiable}, all based on the Williams and Beer PID framework and have been studied further~\cite{makkeh2017bivariate, makkeh2018optimizing, rauh2017extractable,barrett2015exploration, pica2017quantifying}. Despite the importance of the PID to characterise the distribution of information, an agreement on which of these measures is the most suitable is still lacking~\cite{bertschinger2013shared, olbrich2015information}. Hence, on this paper we will only use some basic identities of PID (the general Williams and Beer framework) when building the mapping between PID and Fourier representations.

The main identities of PID are simple and follow from the nature of the decomposition. The first identity states that the joint mutual information $\MI(\RV{T};\RV{X}_1,\dots,\RV{X}_n)$ is the sum of all the PID terms. In the bivariate case, when the system is $(\RV{T},\RV{X}_1,\RV{X}_2)$, this identity reads:

\begin{equation}\label{eq:pid-id-1}
  \MI(\RV{T};\RV{X}_1,\RV{X}_2)= \UI(\RV{T},\RV{X}_1\backslash\RV{X}_2) + \UI(\RV{T};\RV{X}_2\backslash\RV{X}_1) + \CI(\RV{T};\RV{X}_1:\RV{X}_2) + \SI(\RV{T};\RV{X}_1,\RV{X}_2) \, , 
\end{equation}
        
\noindent where $\UI(\RV{T};\RV{X}_1\backslash\RV{X}_2)$ (resp.~$\UI(\RV{T};\RV{X}_2\backslash\RV{X}_1)$) is the information that $\RV{X}_1$ (resp.~$\RV{X}_2$) holds uniquely about $\RV{T}$, $\CI(\RV{T};\RV{X}_1:\RV{X}_2)$ is the information that $\RV{X}_1$ and $\RV{X}_2$ hold synergistically about $\RV{T}$, $\SI(\RV{T};\RV{X}_1,\RV{X}_2)$ is the information that $\RV{X}_1$ and $\RV{X}_2$ hold redundantly about $\RV{T}$. 
        
The other identities state that the information that a specific source has mutually about the target is the sum of the unique and redundant information of this target. In the bivariate case these identities are given by:

\begin{equation}\label{eq:pid-id-2}
  \MI(\RV{T};\RV{X}_i)= \UI(\RV{T},\RV{X}_i\backslash\RV{X}_j) + \SI(\RV{T};\RV{X}_i,\RV{X}_j)\quad\forall~i,j\in\{1,2\}.  
\end{equation}

	% % % % % % % %
	% % % % % % % %
	%
\section{Mapping PID terms to Fourier Coefficients}~\label{sec:the-mapping}
	    
The aim of this section is to introduce a relation between the PID terms and the Fourier transform of Boolean functions. We start by describing the general idea for finding the relation, and then compute the mapping for the bivariate and trivariate cases.
	    
\subsection{The Mapping Scheme}
Suppose $f:\{-1,1\}^n\to\{-1,1\}$ is the mechanism of a multivariate Boolean system or a gate $(\RV{T},\RV{X}_1,\dots,\RV{X}_n)$, where  $\RV{T}=f(\RV{X}_1,\dots,\RV{X}_n)$ and the inputs are equipped with a certain distribution.  Heckel et al.~\cite{heckel2013harmonic} constructed a mapping $\Psi$ from the Fourier coefficients of the mechanism $f:\{-1,1\}^n\to\{-1,1\}$ to the mutual information between components $(\RV{T},\RV{X}_1,\dots,\RV{X}_n)$ (see subsection~\ref{subsec:heckel}). 

Our goal is to construct a map $\Phi$ from the PID terms to the Fourier coefficients of a Boolean mechanism. See Figure~\ref{fig:pid_dia} for a diagram of relations of several mappings between information theoretic functionals and Fourier coefficients used to build $\Phi$.

The mapping $\Psi$ is highly dimensional and nonlinear but it implies an identity map $\operatorname{id}_{\Inf}$ from the set of conditional mutual informations $\{\MI(f(\RV{X}_1,\dots,\RV{X}_n);\RV{X}_i\mid \RV{A}_i) \mid i\in[n]\}$, where $\RV{A}_i= (\RV{X}_1,\dots,\RV{X}_{i-1},\RV{X}_{i+1}, \dots, \RV{X}_n)$, to the set of influences of variables $\{\Inf_i\mid \forall i\in[n]\}$.
    	
Thus, to construct the mapping $\Phi$, one would like to use in addition to $\operatorname{id}_{\Inf}$, the linear map $F_d$ from the PID terms to $\{\MI(f(\RV{X}_1,\dots,\RV{X}_n);\RV{X}_i\mid \RV{A}_i) \mid i\in[n]\}$. Then, a map from PID to Fourier coefficients could be obtained from the PID identities and the linear map $F_r$ from $\{\hat{f}^2_i \mid i\subseteq[n]\backslash\emptyset\}$ to $\{\Inf_i\mid \forall i\in[n]\}$ defined by~\cite[Theorem 2.20]{o2014analysis}.  
    	
Therefore, $\Phi = F_r^{-1}\circ\operatorname{id}_{\Inf}\circ F_d$ is a mapping from the PID terms to $\{\hat{f}^2_i \mid i\subseteq[n]\backslash\emptyset\}$. Figure~\ref{fig:pid_dia} presents the diagram of relations explained above between the information and Fourier components of Boolean gates. Importantly, note that the motivation of the particular choice of building $\Phi$ is the high dimentionality of the PID space (its dimension grows super exponentially). Constructing the intended mapping using the inverse of $\Psi$ which is in its turn high dimensional and nonlinear deems the mapping too complicated to be exploited. 
    	
        \begin{figure}[!t]
            \centering
            \includegraphics[width=12cm, height=8cm]{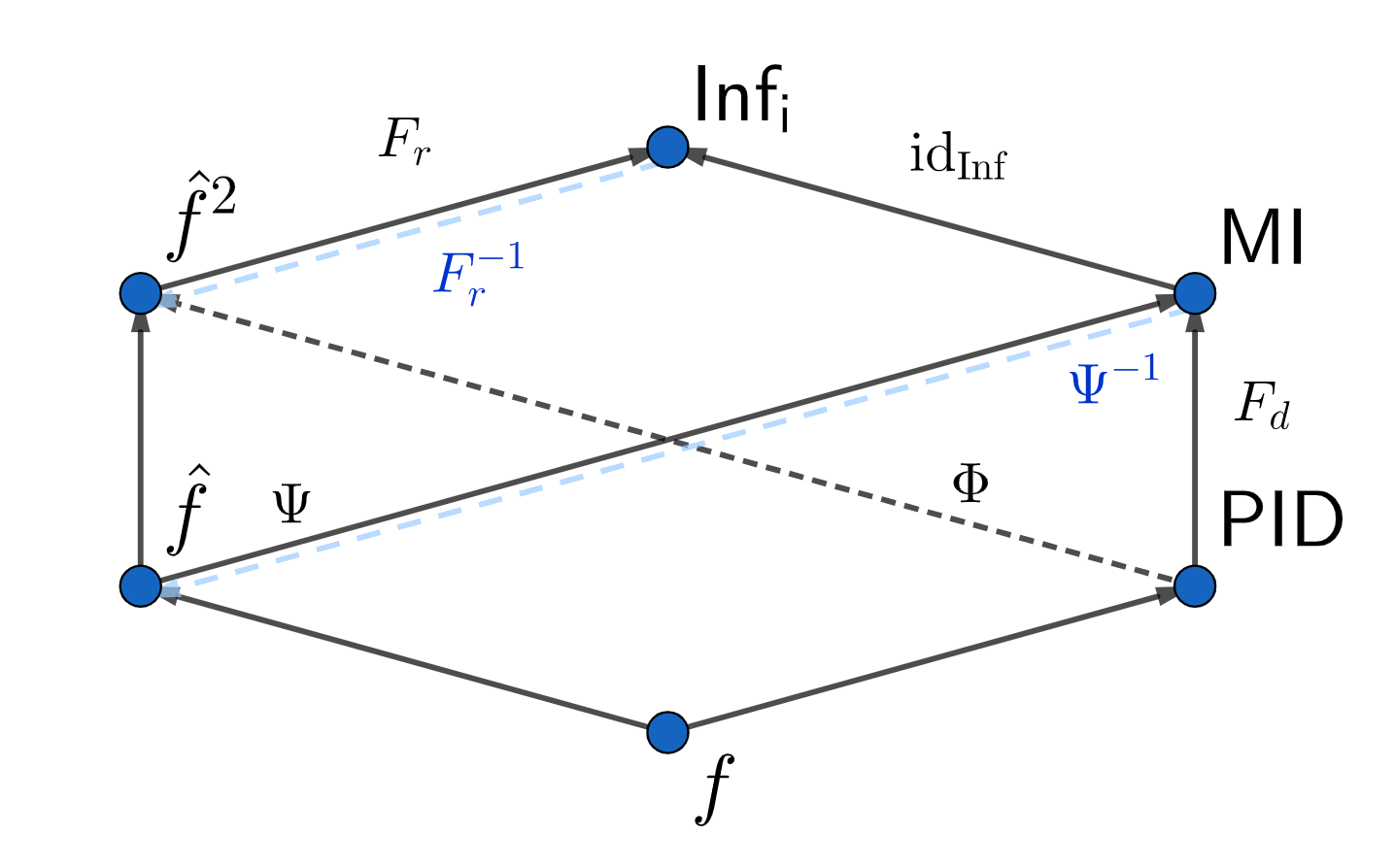}
            \caption{\label{fig:pid_dia} Relations between information theoretic functionals (including PID) and Fourier space of Boolean functions. The solid arrows represent mappings that are known and well defined. The light blue dashed arrows represents the inverse of the corresponding solid maps. The dashed arrow is the mapping, introduced in this paper, from PID space to the squared Fourier coefficients space.}
        \end{figure}
    	
The mapping $\Phi$ can allow one to gain a broader understanding of how PID terms relate to basic properties of Boolean gates. In particular, it enables obtaining new properties of PID by employing well explored and defined properties and relations of the Fourier space.

% % % % % % %
\subsection{Bivariate Systems}
Let $(\RV{T},\RV{X},\RV{Y})$ be a bivariate system and $P$ be the joint probability of $(\RV{T},\RV{X},\RV{Y})$. This system is a Boolean gate if for any $x,y\in X\times Y$ there exists one and only one $t\in T$ such that $P(t,x,y)>0$ and $X=Y=T=\{-1,1\}$.

\begin{theorem}\label{thm:fou-pid-bi}
    Let $f:\{-1,1\}^2\to\{-1,1\}$ be a Boolean function where $\mathfrak{F}'$ the space of squared Fourier coefficients of $f$, $\RV{T} = f(\RV{X},\RV{Y})$, and $\mathfrak{D}$ is the PID of the system $(\RV{T},\RV{X},\RV{Y})$. Then there is a mapping $\Phi = \mathfrak{D}\to\mathfrak{F}'$ such that
	\begin{equation}\label{eq:fou-pid-bi-exp}
	    \begin{split}
		    \hat{f}(\{X,Y\})^2 	&= 2\CI(\RV{T};\RV{X}:\RV{Y}) + \UI(\RV{T};\RV{X}\backslash\RV{Y}) + \UI(\RV{T};\RV{Y}\backslash\RV{X}) + \Exp[f]^2 - 1,\\ 
		    \hat{f}(\{X\})^2	&= 1 - \CI(\RV{T};\RV{X}:\RV{Y}) - \UI(\RV{T};\RV{Y}\backslash\RV{X}) - \Exp[f]^2,\\ 
		    \hat{f}(\{Y\})^2 	&= 1 - \CI(\RV{T};\RV{X}:\RV{Y}) - \UI(\RV{T};\RV{X}\backslash\RV{Y}) - \Exp[f]^2.\\
	    \end{split}				
    \end{equation}
\end{theorem}
\noindent The proof of Theorem~\ref{thm:fou-pid-bi} is deferred to Appendix~\ref{apx:sec:proof:thm:fou-pid-bi}.

\subsubsection{Monotonicity}
Monotone Boolean functions exhibit a special relation between their influence and Fourier coefficients. This relation makes it easier to extract the mapping from PID to the space of Fourier coefficients.  

                \begin{proposition}\label{prop:fou-pid-bi-mono}
			        Let $f:\{-1,1\}^2\to\{-1,1\}$ be a monotone Boolean function where $\mathfrak{F}$ the space of Fourier coefficients of $f$, $\RV{T} = f(\RV{X},\RV{Y})$, and $\mathfrak{D}$ is the PID of the system $(\RV{T},\RV{X},\RV{Y})$. Then there is a mapping $\Phi = \mathfrak{D}\to\mathfrak{F}$ such that  
			        \begin{equation*}
    					\begin{split}
    						\hat{f}(\{x\}) &=   \begin{cases}
    						                        \CI(\RV{T};\RV{X}:\RV{Y}) + \UI(\RV{T};\RV{X}\backslash\RV{Y}) &~\text{if $f(\RV{X},\RV{Y})$ is increasing}\\
    						                        -\CI(\RV{T};\RV{X}:\RV{Y}) - \UI(\RV{T};\RV{X}\backslash\RV{Y}) &~\text{if $f(\RV{X},\RV{Y})$ is decreasing}
    						\end{cases},\\
    						\hat{f}(\{y\}) &=   \begin{cases}
    						                        \CI(\RV{T};\RV{X}:\RV{Y}) + \UI(\RV{T};\RV{Y}\backslash\RV{X}) &~\text{if $f(\RV{X},\RV{Y})$ is increasing}\\
    						                        -\CI(\RV{T};\RV{X}:\RV{Y}) - \UI(\RV{T};\RV{Y}\backslash\RV{X}) &~\text{if $f(\RV{X},\RV{Y})$ is decreasing}
    						                    \end{cases},\\
    						\hat{f}^2(\{x,y\}) &=   \hat{f}(\{x\}) (1 - \hat{f}(\{x\})),\\
    						\hat{f}^2(\{x,y\}) &=   \hat{f}(\{y\}) (1 - \hat{f}(\{y\})).
    					\end{split}
				    \end{equation*}
			    \end{proposition}
			    \begin{proof}
			       The mapping can be easily concluded from Proposition~\ref{prop:mono-inf} and the definition of influence. 
			    \end{proof}
				Proposition~\ref{prop:fou-pid-bi-mono} leads to a consistency equation for monotone functions (gates):
				\begin{equation}
					\begin{split}
						&(\CI(\RV{T};\RV{X}:\RV{Y}) + \UI(\RV{T};\RV{X}\backslash\RV{Y}))(1 - \CI(\RV{T};\RV{X}:\RV{Y}) - \UI(\RV{T};\RV{X}\backslash\RV{Y}))\\
						=\\
						&(\CI(\RV{T};\RV{X}:\RV{Y}) + \UI(\RV{T};\RV{Y}\backslash\RV{X}))(1 - \CI(\RV{T};\RV{X}:\RV{Y}) - \UI(\RV{T};\RV{Y}\backslash\RV{X})).
					\end{split}
				\end{equation}
				This equation might be limited due to the small number of monotone bivariate functions (only six not counting negations) but it is significant in the $p$-biased case since the number of $p$-biased monotone bivariate functions is in principle infinite.
				
				The subsection is concluded by expressing the mapping for unate functions. This family functions of functions includes the linear threshold functions and they have applications in biological networks~\cite{raeymaekers02, grefenstette06}.
                \begin{corollary}
                    Let $f:\{-1,1\}^2\to\{-1,1\}$ be unate function where $\mathfrak{F}$ the space of Fourier coefficients of $f$, $\RV{T} = f(\RV{X},\RV{Y})$, $\mathfrak{D}$ is the PID of the system $(\RV{T},\RV{X},\RV{Y})$, and $(a_1,a_2)$ are its unate parameters. Then there is a mapping $\Phi = \mathfrak{D}\to\mathfrak{F}$ such that
    			        \begin{equation*}
        					\begin{split}
        						\hat{f}(\{x\}) &=   a_1(\CI(\RV{T};\RV{X}:\RV{Y}) + \UI(\RV{T};\RV{X}\backslash\RV{Y}) ),\\
        						\hat{f}(\{y\}) &=   a_2(\CI(\RV{T};\RV{X}:\RV{Y}) + \UI(\RV{T};\RV{Y}\backslash\RV{X}) ),\\
        						\hat{f}^2(\{x,y\}) &=   \hat{f}(\{x\}) (1 - \hat{f}(\{x\})),\\
        						\hat{f}^2(\{x,y\}) &=   \hat{f}(\{y\}) (1 - \hat{f}(\{y\})).
        					\end{split}
    				    \end{equation*}
                \end{corollary}
                \noindent Using Proposition~\ref{prop:unate}, the proof of the above corollary follows the lines of the proof of Proposition~\ref{prop:fou-pid-bi-mono}.
		\subsection{Trivariate Systems}
			Let $(\RV{T},\RV{X},\RV{Y},\RV{Z})$ be a trivariate system and $P$ the joint probability of $(\RV{T},\RV{X},\RV{Y},\RV{Z})$. This system is a Boolean gate if for any $x,y,z\in X\times Y\times Z$ there exists one and only one $t\in T$ such that $P(t,x,y,z)>0$ and $T=X=Y=Z=\{-1,1\}$.  
		
		    \begin{proposition}\label{prop:fou-pid-tri}
		        Let $f:\{-1,1\}^3\to\{-1,1\}$ be a Boolean function where $\mathfrak{F}'$ the space of squared Fourier coefficients of $f$, $\RV{T} = f(\RV{X},\RV{Y},\RV{Z})$, and $\mathfrak{D}$ is the PID of the system $(\RV{T},\RV{X},\RV{Y},\RV{Z})$. Then there is a mapping $\Phi = \mathfrak{D}\to\RR$ such that		        			\begin{equation}
    				\begin{aligned}
    					\Phi_0(\mathfrak{D}) - \nfrac{2}{8} &\le \hat{f}(\{X,Y,Z\})^2	&\le \Phi_0(\mathfrak{D}) + \nfrac{5}{8},\quad \Phi_1(\mathfrak{D}) - \nfrac{2}{8} &\le \hat{f}(\{X,Y\})^2 	&\le \Phi_1(\mathfrak{D}) + \nfrac{4}{8},\\
    					\Phi_2(\mathfrak{D}) - \nfrac{2}{8} &\le \hat{f}(\{X,Z\})^2		&\le \Phi_2(\mathfrak{D}) + \nfrac{4}{8},\quad \Phi_3(\mathfrak{D}) - \nfrac{2}{8} &\le \hat{f}(\{Y,Z\})^2	&\le \Phi_3(\mathfrak{D}) + \nfrac{4}{8}\\
    					\Phi_4(\mathfrak{D}) - \nfrac{2}{8} &\le \hat{f}(\{X\})^2		&\le \Phi_4(\mathfrak{D}) + \nfrac{5}{8},\quad \Phi_5(\mathfrak{D}) - \nfrac{2}{8} &\le \hat{f}(\{Y\})^2	&\le \Phi_5(\mathfrak{D}) + \nfrac{5}{8}\\
    					\Phi_6(\mathfrak{D}) - \nfrac{2}{8} &\le \hat{f}(\{Z\})^2		&\le \Phi_6(\mathfrak{D}) + \nfrac{5}{8}.\qquad\qquad\qquad
    				\end{aligned}
			    \end{equation}
                where $\Phi = (\Phi_0,\dots,\Phi_6)$\footnote{The mappings $\Phi_i$ for all $i$ are in the Appendix~\ref{apx:sec:apx-maps}, see~\eqref{eq:apx-phi-tri-1} and~\eqref{eq:apx-phi-tri-2}.}. 
		    \end{proposition}
			\noindent The proof of Proposition~\ref{prop:fou-pid-tri} is deferred to Appendix~\ref{apx:sec:proof:prop:fou-pid-tri}. This proposition does not prove that there is a unique of mapping $\Phi:\mathfrak{D}\to\mathfrak{F}'$ but rather a family of mappings. 
			
			As no proof for such mapping was formulated, a check was performed over all Boolean functions $f:\{-1,1\}^3\to\{-1,1\}$ with uniformly distributed inputs, for the existence of two Boolean functions such that their PID are identical whereas their squared Fourier coefficients are different. Note that the PID quantities were derived according to the maximum-entropy measure introduced in~\cite{Chicharro17} and computed using the algorithm \textsc{MaxEnt3D\_Pid} described in~\cite{makkeh2019maxent3d_pid}. As no contradicting example was found in the above experiments, this motivates the following conjecture.
			\begin{conjecture}\label{conj:fou-pid-tri}
			    Let $f:\{-1,1\}^3\to\{-1,1\}$ be a Boolean function where $\mathfrak{F}'$ the space of squared Fourier coefficients of $f$, $\RV{T} = f(\RV{X},\RV{Y},\RV{Z})$, and $\mathfrak{D}$ is the PID of the system $(\RV{T},\RV{X},\RV{Y},\RV{Z})$. Then there exists a mapping $\Phi = \mathfrak{D}\to\mathfrak{F}'.$
			\end{conjecture}
            \subsubsection{Monotonicity}
                Similarly to the bivariate case, Proposition~\ref{prop:mono-inf} can be exploited in order to find the mapping between the PID and Fourier coefficients $\Phi:\mathfrak{D}\to\mathfrak{F}'$ for any trivariate monotone Boolean function. 
                \begin{theorem}\label{thm:fou-pid-tri-mono}
			        Let $f:\{-1,1\}^3\to\{-1,1\}$ be a monotone Boolean function where $\mathfrak{F}$ the space of Fourier coefficients of $f$, $\RV{T} = f(\RV{X},\RV{Y},\RV{Z})$, and $\mathfrak{D}$ is the PID of the system $(\RV{T},\RV{X},\RV{Y},\RV{Z})$. Then there is a mapping $\Phi = \mathfrak{D}\to\mathfrak{F}$ such that
			        \begin{equation*}
    					\begin{split}
    					    \hat{f}^2(\{x,y,z\})    &= 2\lt(\psi_0(\psi_0 + \nfrac{1}{2}) + \psi_1(\psi_1 + \nfrac{1}{2}) + \psi_2(\psi_2 + \nfrac{1}{2}) +\Exp[f]^2 - 1)\rt),\\
    					    \hat{f}^2(\{x,y\})      &= 1 -\Exp[f]^2 -\psi_0^2 - \psi_1^2 - \psi_2(1+\psi_2),\\
    						\hat{f}^2(\{x,z\})      &= 1 -\Exp[f]^2 -\psi_0^2 - \psi_2^2 - \psi_1(1+\psi_1),\\
    						\hat{f}^2(\{y,z\})      &= 1 -\Exp[f]^2 -\psi_1^2 - \psi_2^2 - \psi_0(1+\psi_0),\\
    						\hat{f}(\{x\})          &=  \begin{cases}
    						                                \psi_0  &~\text{if $f$ is increasing}\\
    						                                -\psi_0 &~\text{if $f$ is decreasing}
    						                            \end{cases},\\
    						\hat{f}(\{y\})          &= \begin{cases}
    						                                \psi_1  &~\text{if $f$ is increasing}\\
    						                                -\psi_1 &~\text{if $f$ is decreasing}
    						                            \end{cases},\\
    						\hat{f}(\{z\})          &= \begin{cases}
    						                                \psi_2  &~\text{if $f$ is increasing}\\
    						                                -\psi_2 &~\text{if $f$ is decreasing}
    						                            \end{cases},
    					\end{split}
				    \end{equation*}
				    where
			        \begin{equation*}
			            \begin{split}
			                \psi_0  &= \CI(\RV{T};\RV{X}:\RV{Y}:\RV{Z}) + \CI(\RV{T};\RV{X}:\RV{Y}) + \CI(\RV{T};\RV{X}:\RV{Z}) + \CI(\RV{T};\RV{X}:\RV{Y},\RV{X}:\RV{Z})\\                  &+\UI(\RV{T};\RV{X}\backslash\RV{Y},\RV{Z}),\\
			                \psi_1  &= \CI(\RV{T};\RV{X}:\RV{Y}:\RV{Z}) + \CI(\RV{T};\RV{X}:\RV{Y}) + \CI(\RV{T};\RV{Y}:\RV{Z}) + \CI(\RV{T};\RV{X}:\RV{Y},\RV{Y}:\RV{Z})\\
    						        &+\UI(\RV{T};\RV{Y}\backslash\RV{X},\RV{Z}),\\
    						\psi_2  &= \CI(\RV{T};\RV{X}:\RV{Y}:\RV{Z}) + \CI(\RV{T};\RV{X}:\RV{Z}) + \CI(\RV{T};\RV{Y}:\RV{Z}) + \CI(\RV{T};\RV{X}:\RV{Z},\RV{Y}:\RV{Z})\\
    						        &+ \UI(\RV{T};\RV{Z}\backslash\RV{X},\RV{Y}).
			            \end{split}
			        \end{equation*}
                \end{theorem}
                \noindent In contrast to the bivariate case the mapping is quadratic in the PID terms and hence nonlinear. The proof of Theorem~\ref{thm:fou-pid-tri-mono} is deferred to the Appendix~\ref{apx:sec:proof:thm:fou-pid-tri-mono}. Finally, the mapping for unate functions is shown.
                \begin{corollary}\label{coro:fou-pid-tri-unate}
			        Let $f:\{-1,1\}^3\to\{-1,1\}$ be unate function where $\mathfrak{F}$ the space of Fourier coefficients of $f$, $\RV{T} = f(\RV{X},\RV{Y},\RV{Z})$, $\mathfrak{D}$ is the PID of the system $(\RV{T},\RV{X},\RV{Y},\RV{Z})$, and $(a_1,a_2,a_3)$ are its unate parameters. Then there is a mapping $\Phi = \mathfrak{D}\to\mathfrak{F}$ such that
			        \begin{equation*}
    					\begin{split}
    					    \hat{f}^2(\{x,y,z\})    &= 2\lt(\psi_0(\psi_0 + \nfrac{1}{2}) + \psi_1(\psi_1 + \nfrac{1}{2}) + \psi_2(\psi_2 + \nfrac{1}{2}) +\Exp[f]^2 - 1)\rt),\\
    					    \hat{f}^2(\{x,y\})      &= 1 -\Exp[f]^2 -\psi_0^2 - \psi_1^2 - \psi_2(1+\psi_2),\\
    						\hat{f}^2(\{x,z\})      &= 1 -\Exp[f]^2 -\psi_0^2 - \psi_2^2 - \psi_1(1+\psi_1),\\
    						\hat{f}^2(\{y,z\})      &= 1 -\Exp[f]^2 -\psi_1^2 - \psi_2^2 - \psi_0(1+\psi_0),\\
    						\hat{f}(\{x\})          &= \psi_0,\\
    						\hat{f}(\{y\})          &= \psi_1,\\
    						\hat{f}(\{z\})          &= \psi_2,
    					\end{split}
				    \end{equation*}
				    where
			        \begin{equation*}
			            \begin{split}
			                \psi_0  &= a_1( \CI(\RV{T};\RV{X}:\RV{Y}:\RV{Z}) + \CI(\RV{T};\RV{X}:\RV{Y}) + \CI(\RV{T};\RV{X}:\RV{Z}) + \CI(\RV{T};\RV{X}:\RV{Y},\RV{X}:\RV{Z})\\                  &+\UI(\RV{T};\RV{X}\backslash\RV{Y},\RV{Z}) ),\\
			                \psi_1  &= a_2( \CI(\RV{T};\RV{X}:\RV{Y}:\RV{Z}) + \CI(\RV{T};\RV{X}:\RV{Y}) + \CI(\RV{T};\RV{Y}:\RV{Z}) + \CI(\RV{T};\RV{X}:\RV{Y},\RV{Y}:\RV{Z})\\
    						        &+\UI(\RV{T};\RV{Y}\backslash\RV{X},\RV{Z}) ),\\
    						\psi_2  &= a_3( \CI(\RV{T};\RV{X}:\RV{Y}:\RV{Z}) + \CI(\RV{T};\RV{X}:\RV{Z}) + \CI(\RV{T};\RV{Y}:\RV{Z}) + \CI(\RV{T};\RV{X}:\RV{Z},\RV{Y}:\RV{Z})\\
    						        &+ \UI(\RV{T};\RV{Z}\backslash\RV{X},\RV{Y}) ).
			            \end{split}
			        \end{equation*}
                \end{corollary}
                \noindent Using Proposition~\ref{prop:unate}, the proof of the above corollary follows the lines of the proof of Theorem~\ref{thm:fou-pid-tri-mono}.
% % % %	%	
% % % % %
% % % % %
	\section{Application for Bivariate Boolean Functions}\label{sec:apps}
	    Fourier analysis of Boolean functions is extensively studied throughout the literature. Many important properties described in terms of their Fourier coefficients can be translated to the PID terms via the mapping~$\Phi$. Ultimately, these well understood properties when viewed from the PID terms should broaden our understanding of these terms and may reveal some novel interpretations. 
	    
	    This section serves as a study case for the application of the mapping between PID to the Fourier space to describe some properties of bivariate Boolean mechanisms from the lens of PID. In particular, here we focus on the contributions to sensitivity and stability by the different PID terms.
	    
		\subsection{Influence and PID}
		    Recall from Proposition~\ref{prop:tot-inf} that the total influence of a Boolean function $f$ is its average sensitivity to perturbation of its inputs. Then, we start by looking at how the different PID terms contribute to the total influence and hence to the sensitivity of a function. Using~\eqref{eq:fou-pid-bi-exp}, the total influence can be expressed as 
			\begin{equation}\label{eq:fou-pid-bi-Inf}
				\Inf[f] = 2\CI(\RV{T};\RV{X}:\RV{Y}) + \UI(\RV{T};\RV{X}\backslash\RV{Y}) + \UI(\RV{T};\RV{Y}\backslash\RV{X}).
			\end{equation}
			
			Thus, the influence or sensitivity of a function is equal to twice the synergy between the sources plus the each of the unique information terms. This is in line with a desired and intuitive property of shared information, namely that the higher the fraction of shared information of $\RV{X}$ and $\RV{Y}$ among the mutual information between the sources and the target, the more robust the function should be to perturbations. Besides, the synergistic information plays a significant adversary role in the sensitivity of the function compared to unique information. Moreover, note that for bivariate functions the contributions of the PID terms to the total influence are linear. This is not generally the case for multivariate functions (for example trivariate monotone functions have a quadratic relation between their influence and PID terms).

            We also note that the total influence is lower bounded by Theorem~\ref{thm:tot-inf} and hence we can directly obtain a lower bound for certain combinations of PID terms. In the bivariate case this reads
			\begin{equation}\label{eq:Poincare-PID}
			    \begin{split}
			        2\alpha\log(\frac{1}{\alpha})   &\leq 2\CI(\RV{T};\RV{X}:\RV{Y}) + \UI(\RV{T};\RV{X}\backslash\RV{Y}) + \UI(\RV{T};\RV{Y}\backslash\RV{X})\\
			                                        &\le \MI(\RV{T};\RV{X},\RV{Y}) - \CoI(\RV{T};\RV{X};\RV{Y}),
			    \end{split}
			\end{equation}
			where $\alpha = \min\{P[f(\RV{X},\RV{Y})=1],P[f(\RV{X},\RV{Y})=-1]\}$ and 
			\[
			    \CoI(\RV{T};\RV{X};\RV{Y}):= \MI(\RV{T};\RV{X}) - \MI(\RV{T};\RV{X}\mid\RV{Y}) = \MI(\RV{T};\RV{Y}) - \MI(\RV{T};\RV{Y}\mid\RV{X}).  
			\]
			The quantity $\CoI(\RV{T};\RV{X};\RV{Y})$ is called the co-information --- sometimes referred to as the interaction information --- which is the multivariate mutual information of $\RV{T},$ $\RV{X},$ and $\RV{Y}.$ Using the PID identities~\eqref{eq:pid-id-1} and~\eqref{eq:pid-id-2}, $\CoI(\RV{T};\RV{X};\RV{Y}) = \SI(\RV{T};\RV{X},\RV{Y}) - \CI(\RV{T};\RV{X}:\RV{Y})$ and in fact long before the Williams-Beer PID framework, the sign of co-information was thought of as an indication to whether the system interacts in a synergistic or redundant manner. 
			
			Inequality \eqref{eq:Poincare-PID} implies that a lower bound for twice synergy plus unique terms is maximized for functions with imbalance in their distribution of output bits.  
		
		\subsection{Noise sensitivity and PID}
		    In Subsection~\ref{subsec:bg-fou}, the noise stability $\Stab_\rho [f]$ of a Boolean function $f$ is said to measure the correlation between $f(\RV{X}_1,\RV{Y}_1)$ and $f(\RV{X}_2,\RV{Y}_2)$ when $((\RV{X}_1,\RV{Y}_1),(\RV{X}_2,\RV{Y}_2))$ is a $\rho$-correlated pair. In other words, the noise stability measures the correlations between two target values when their source values are $\rho$-correlated. 
		    
		    Using the mapping $\Phi$ in~\eqref{eq:fou-pid-bi-exp} and the definition~\eqref{eq:stab}, the stability of Boolean functions at $\rho\in[-1,1]$ can be expressed via the PID terms as 
			\begin{equation}
				\begin{split}
					\Stab_\rho [f]	&= (\rho^2 - \rho)(2\CI(\RV{T};\RV{X}:\RV{Y}) + \UI(\RV{T};\RV{X}\backslash\RV{Y}) + \UI(\RV{T};\RV{Y}\backslash\RV{X})) + (\rho^2 - 2\rho)(\Exp[f]^2 - 1).%,\\
				\end{split}
			\end{equation}
			
			Recall that another way to look at the robustness of the Boolean function is via its noise sensitivity at some $\delta\in\{0,1\}$. In particular, if every source is flipped with a probability $\delta$ then $NS_\delta[f]$ is the probability that the value of $f$ flips. The noise sensitivity of $f$ at $\delta\in\{0,1\}$ can be written in terms of the PID terms as 
			\begin{equation}
				\begin{split}
					\NS_\delta [f]	&= 2\delta^2 + (1 - 2\delta)\lt(\delta(2\CI(\RV{T};\RV{X}:\RV{Y}) + \UI(\RV{T};\RV{X}\backslash\RV{Y}) + \UI(\RV{T};\RV{Y}\backslash\RV{X})) + \frac{(1 + 2\delta)}{2}\Exp[f]^2 \rt).
				\end{split}
			\end{equation}
			
			As with the case of influence, the same combination of PID terms appears in either the stability or the the noise sensitivity. However, now the PID terms appear also in combination with the bias of the function  $\Exp[f]^2$. In this case, manipulating the PID terms will in general also affect the bias of the resulting function, and hence the trade-off between synergy and unique information with the bias of the resulting function determines how the stability and noise sensitivity is changed. 
	% % % % %
	% % % % %
	% % % % %
	%
	\section{Boolean Functions with Biased Inputs}\label{sec:p-biased}
	    So far all the functions that we have considered were equipped with a uniform distribution for their input (sources). Next, we establish the mapping  $\Phi:\mathfrak{D}\to\mathfrak{F}$ for the case where the input distribution is biased.
	    
		Let $f:\{-1,1\}^n\to\{-1,1\}$ be a Boolean function. The value of the Boolean function, $f(\RV{X}_1,\dots,\RV{X}_n),$ is the target $\RV{T}$ of the sources $\RV{X}_1,\dots,\RV{X}_n$. Let $P'$ be the distribution of $(\RV{X}_1,\dots,\RV{X}_n)$ where the variance of $\RV{X}_i$ is denoted by $\Var[\RV{X}_i]$, its standard deviation is denoted by $\sigma_i$, and its expectation is denoted by $\mu_i.$
		
		In the following, the biased bit concept is the case of interest, i.e., each $\RV{X}_i$ is equal to $1$ with probability $p_i$. For simplicity, it is assumed that all $p_i$ are equal to $p$, i.e., in the p-biased model the mean and standard deviation are the same for each bit $i$; $\mu_i = \mu = 2p - 1$ and standard deviation $\sigma_i = \sigma = 2\sqrt{p}\sqrt{1 - p}$. For more details on the Fourier analysis of the $p$-biased models we refer to Appendix~\ref{apx:p-biased-bg}.
		% % % % % % % % %
		% % % % % % % % %
		% %
		\subsection{Bivariate Case}
		    Let $f:\{-1,1\}^2\to\{-1,1\}$ be a $p$-biased bivariate Boolean function. The following theorem shows that there exists a mapping from the PID space of $f$ to its Fourier coefficient space. 
			\begin{theorem}\label{thm:fou-pid-bi-p}
			    Let $f:\{-1,1\}^2\to\{-1,1\}$ be a $p$-biased Boolean function where $\mathfrak{F}'$ the space of squared Fourier coefficients of $f$, $\RV{T} = f(\RV{X},\RV{Y})$, and $\mathfrak{D}$ is the PID of the system $(\RV{T},\RV{X},\RV{Y})$. Then there is a mapping $\Phi:\mathfrak{D}\to\mathfrak{F}'$ such that
			    \begin{equation}\label{eq:fou-pid-bi-p-exp}
				    \begin{split}
    					\hat{f}(\{X,Y\})^2 	&= \frac{\sigma}{h(p)}\bigl(2\CI(\RV{T};\RV{X}:\RV{Y}) + \UI(\RV{T};\RV{X}\backslash\RV{Y}) + \UI(\RV{T};\RV{Y}\backslash\RV{X})\bigr) + \Exp[f]^2 - 1,\\ 
    					\hat{f}(\{X\})^2	&= 1 - \frac{\sigma}{h(p)}\bigl(\CI(\RV{T};\RV{X}:\RV{Y}) + \UI(\RV{T};\RV{Y}\backslash\RV{X})\bigr) - \Exp[f]^2,\\ 
    					\hat{f}(\{Y\})^2 	&= 1 - \frac{\sigma}{h(p)}\bigl(\CI(\RV{T};\RV{X}:\RV{Y}) + \UI(\RV{T};\RV{X}\backslash\RV{Y})\bigr) - \Exp[f]^2.
				    \end{split}				
			    \end{equation}
			\end{theorem}
			 \noindent The proof of Theorem~\ref{thm:fou-pid-bi-p} is deferred to Appendix~\ref{apx:sec:proof:thm:fou-pid-bi-p}. In case the function is also monotone, the map $\Phi:\mathfrak{D}\to\mathfrak{F}$ takes the special form shown in the following proposition.
			 
		 	\begin{proposition}
		        Let $f:\{-1,1\}^2\to\{-1,1\}$ be a monotone $p$-biased Boolean function where $\mathfrak{F}$ the space of Fourier coefficients of $f$, $\RV{T} = f(\RV{X},\RV{Y})$, and $\mathfrak{D}$ is the PID of the system $(\RV{T},\RV{X},\RV{Y})$. Then there is a mapping $\Phi: \mathfrak{D}\to\mathfrak{F}$ such that
		        \begin{equation*}
					\begin{split}
						\hat{f}(\{x\}) &=\begin{cases} 
						                        \frac{\sigma}{H(\RV{X})}\lt(\CI(\RV{T};\RV{X}:\RV{Y}) + \UI(\RV{T};\RV{X}\backslash\RV{Y})\rt)  &~\text{if $f$ is increasing}\\
						                        -\frac{\sigma}{H(\RV{X})}\lt(\CI(\RV{T};\RV{X}:\RV{Y}) + \UI(\RV{T};\RV{X}\backslash\RV{Y})\rt) &~\text{if $f$ is decreasing}
						                  \end{cases},\\
						\hat{f}(\{y\}) &= \begin{cases}
						                        \frac{\sigma}{H(\RV{Y})}\lt(\CI(\RV{T};\RV{X}:\RV{Y}) + \UI(\RV{T};\RV{Y}\backslash\RV{X})\rt)  &~\text{if $f$ is increasing}\\
						                        -\frac{\sigma}{H(\RV{Y})}\lt(\CI(\RV{T};\RV{X}:\RV{Y}) + \UI(\RV{T};\RV{Y}\backslash\RV{X})\rt) &~\text{if $f$ is decreasing}\\
						                  \end{cases},\\
						\hat{f}^2(\{x,y\}) &= \hat{f}(\{x\}) (1 - \hat{f}(\{x\})),\\
						\hat{f}^2(\{x,y\}) &= \hat{f}(\{y\}) (1 - \hat{f}(\{y\})).
					\end{split}
			    \end{equation*}
		    \end{proposition}
		    \noindent The proof is trivial as the mapping can be easily concluded from the influence of p-biased functions and the definition of influence. The subsection is concluded by expressing the mapping for unate functions.
			 \begin{corollary}
		        Let $f:\{-1,1\}^2\to\{-1,1\}$ be $p$-biased unate where $\mathfrak{F}$ the space of Fourier coefficients of $f$, $\RV{T} = f(\RV{X},\RV{Y})$, $\mathfrak{D}$ is the PID of the system $(\RV{T},\RV{X},\RV{Y})$, and $(a_1,a_2)$ are its unate parameter. Then there is a mapping $\Phi = \mathfrak{D}\to\mathfrak{F}$ such that
		        \begin{equation*}
					\begin{split}
						\hat{f}(\{x\}) &=\frac{a_1\sigma}{H(\RV{X})}\lt(\CI(\RV{T};\RV{X}:\RV{Y}) + \UI(\RV{T};\RV{X}\backslash\RV{Y})\rt),\\
						\hat{f}(\{y\}) &= \frac{a_2\sigma}{H(\RV{Y})}\lt(\CI(\RV{T};\RV{X}:\RV{Y}) + \UI(\RV{T};\RV{Y}\backslash\RV{X})\rt),\\
						\hat{f}^2(\{x,y\}) &= \hat{f}(\{x\}) (1 - \hat{f}(\{x\})),\\
						\hat{f}^2(\{x,y\}) &= \hat{f}(\{y\}) (1 - \hat{f}(\{y\})).
					\end{split}
			    \end{equation*}
		    \end{corollary}
			 
			 Since $H(\RV{X})=H(\RV{Y})$, $p$-biased monotone functions (as those equipped with the uniform measure) admit the following consistency equation $A=B$ where
			 \begin{equation}
			     \begin{split}
			         A  &:=  (\CI(\RV{T};\RV{X}:\RV{Y}) + \UI(\RV{T};\RV{X}\backslash\RV{Y}))(H(\RV{X}) - \CI(\RV{T};\RV{X}:\RV{Y}) - \UI(\RV{T};\RV{X}\backslash\RV{Y})),\\
			         B  &:=  (\CI(\RV{T};\RV{X}:\RV{Y}) + \UI(\RV{T};\RV{Y}\backslash\RV{X}))(H(\RV{Y}) - \CI(\RV{T};\RV{X}:\RV{Y}) - \UI(\RV{T};\RV{Y}\backslash\RV{X})).
			     \end{split}
			 \end{equation}
		% % % % % % % %
		% % % % % % % %
		% %
		\subsection{Trivariate Case}
		    Let $f:\{-1,1\}^3\to\{-1,1\}$ be a $p$-biased trivariate Boolean function. The following proposition shows that certain functions of PID terms lower- and upper-bound different Fourier coefficients. 
		    \begin{proposition}\label{prop:fou-pid-tri-p-biased}
			    Let $f:\{-1,1\}^3\to\{-1,1\}$ be a $p$-biased Boolean function where $\mathfrak{F}'$ the space of squared Fourier coefficients of $f$, $\RV{T} = f(\RV{X},\RV{Y},\RV{Z})$, and $\mathfrak{D}$ is the PID of its system $(\RV{T},\RV{X},\RV{Y},\RV{Z})$. Then there is a mapping $\Phi:\mathfrak{D}\to\mathfrak{F}'$ such that
     			\begin{equation}
     				\begin{aligned}
     					\Phi_0(\mathfrak{D}) - \nfrac{2}{8} &\le\hat{f}(\{X,Y,Z\})^2	&\le \Phi_0(\mathfrak{D}) + \nfrac{5}{8},\quad \Phi_1(\mathfrak{D}) - \nfrac{2}{8} &\le \hat{f}(\{X,Y\})^2 	&\le \Phi_1(\mathfrak{D}) + \nfrac{4}{8},\\
     					\Phi_2(\mathfrak{D}) - \nfrac{2}{8} &\le \hat{f}(\{X,Z\})^2	&\le \Phi_2(\mathfrak{D}) + \nfrac{4}{8},\quad \Phi_3(\mathfrak{D}) - \nfrac{2}{8} &\le \hat{f}(\{Y,Z\})^2	&\le \Phi_3(\mathfrak{D}) + \nfrac{4}{8}\\
     					\Phi_4(\mathfrak{D}) - \nfrac{2}{8} &\le \hat{f}(\{X\})^2	&\le \Phi_4(\mathfrak{D}) + \nfrac{5}{8},\quad \Phi_5(\mathfrak{D}) - \nfrac{2}{8} &\le \hat{f}(\{Y\})^2	&\le \Phi_5(\mathfrak{D}) + \nfrac{5}{8}\\
     					\Phi_6(\mathfrak{D}) - \nfrac{2}{8} &\le \hat{f}(\{Z\})^2	&\le \Phi_6(\mathfrak{D}) + \nfrac{5}{8}.\qquad\qquad\qquad 
     				\end{aligned}
     			\end{equation}
     			where $\Phi = (\Phi_0,\dots,\Phi_6)$\footnote{The mappings $\Phi_i$ for all $i$ are in Appendix~\ref{apx:sec:apx-maps}, see~\eqref{eq:apx-phi-tri-p-biased-1} and~\eqref{eq:apx-phi-tri-p-biased-2}.}.  
			\end{proposition}
            \noindent The proof follows the same lines of that for the uniformly distributed case. In the case that $f$ is a monotone $p$-biased trivariate Boolean function, then the following theorem proves the existence of a quadratic mapping $\Phi:\mathfrak{D}\to\mathfrak{F}$, which can be computed explicitly.
            \begin{theorem}
			        Let $f:\{-1,1\}^3\to\{-1,1\}$ be a monotone $p$-biased Boolean function where $\mathfrak{F}$ the space of Fourier coefficients of $f$, $\RV{T} = f(\RV{X},\RV{Y},\RV{Z})$, and $\mathfrak{D}$ be the PID of the system $(\RV{T},\RV{X},\RV{Y},\RV{Z})$. Then there is a mapping $\Phi = \mathfrak{D}\to\mathfrak{F}$ such that
			        \begin{equation*}
    					\begin{split}
    					    \hat{f}^2(\{x,y,z\})    &= 2\lt(\psi_0(\psi_0 + \nfrac{\sigma}{2H(\RV{X})}) + \psi_1(\psi_1 + \nfrac{\sigma}{2H(\RV{Y})}) + \psi_2(\psi_2 + \nfrac{\sigma}{2H(\RV{Z})}) +\Exp[f]^2 - 1\rt),\\
    					    \hat{f}^2(\{x,y\})      &= 1 -\Exp[f]^2 -\psi_0^2 - \psi_1^2 - \psi_2(\nfrac{\sigma}{H(\RV{Z})}+\psi_2),\\
    						\hat{f}^2(\{x,z\})      &= 1 -\Exp[f]^2 -\psi_0^2 - \psi_2^2 - \psi_1(\nfrac{\sigma}{H(\RV{Y})}+\psi_1),\\
    						\hat{f}^2(\{y,z\})      &= 1 -\Exp[f]^2 -\psi_1^2 - \psi_2^2 - \psi_0(\nfrac{\sigma}{H(\RV{X})}+\psi_0),\\
    						\hat{f}(\{x\})          &=  \begin{cases}
    						                                \psi_0 &~\text{if $f$ is increasing}\\
    						                                -\psi_0 &~\text{if $f$ is decreasing}\\
    						                            \end{cases},\\
    						\hat{f}(\{y\})          &=  \begin{cases}
    						                                \psi_1 &~\text{if $f$ is increasing}\\
    						                                -\psi_1 &~\text{if $f$ is decreasing}\\
    						                            \end{cases},\\
    						\hat{f}(\{z\})          &=  \begin{cases}
    						                                \psi_2 &~\text{if $f$ is increasing}\\
    						                                -\psi_2 &~\text{if $f$ is decreasing}\\
    						                            \end{cases},\\
    					\end{split}
				    \end{equation*}
				    where
			        \begin{equation*}
			            \begin{split}
			                \psi_0  &= \nfrac{\sigma}{H(\RV{X})}\bigl(\CI(\RV{T};\RV{X}:\RV{Y}:\RV{Z}) + \CI(\RV{T};\RV{X}:\RV{Y}) + \CI(\RV{T};\RV{X}:\RV{Z}) + \CI(\RV{T};\RV{X}:\RV{Y},\RV{X}:\RV{Z})\\                  
			                        &+\UI(\RV{T};\RV{X}\backslash\RV{Y},\RV{Z})\bigr),\\
			                \psi_1  &= \nfrac{\sigma}{H(\RV{Y})}\bigl(\CI(\RV{T};\RV{X}:\RV{Y}:\RV{Z}) + \CI(\RV{T};\RV{X}:\RV{Y}) + \CI(\RV{T};\RV{Y}:\RV{Z}) + \CI(\RV{T};\RV{X}:\RV{Y},\RV{Y}:\RV{Z})\\
    						        &+\UI(\RV{T};\RV{Y}\backslash\RV{X},\RV{Z})\bigr),\\
    						\psi_2  &= \nfrac{\sigma}{H(\RV{Z})}\bigl(\CI(\RV{T};\RV{X}:\RV{Y}:\RV{Z}) + \CI(\RV{T};\RV{X}:\RV{Z}) + \CI(\RV{T};\RV{Y}:\RV{Z}) + \CI(\RV{T};\RV{X}:\RV{Z},\RV{Y}:\RV{Z})\\
    						        &+ \UI(\RV{T};\RV{Z}\backslash\RV{X},\RV{Y})\bigr).
			            \end{split}
			        \end{equation*}
                \end{theorem}
                
                Finally, the mapping in the case of unate functions is given by the following corollary.
                \begin{corollary}
			        Let $f:\{-1,1\}^3\to\{-1,1\}$ be $p$-biased unate where $\mathfrak{F}$ the space of Fourier coefficients of $f$, $\RV{T} = f(\RV{X},\RV{Y},\RV{Z})$, $\mathfrak{D}$ is the PID of the system $(\RV{T},\RV{X},\RV{Y},\RV{Z})$, and $(a_0,a_1,a_2)$ are its unate parameters. Then there is a mapping $\Phi = \mathfrak{D}\to\mathfrak{F}$ such that
			        \begin{equation*}
    					\begin{split}
    					    \hat{f}^2(\{x,y,z\})    &= 2\lt(\psi_0(\psi_0 + \nfrac{\sigma}{2H(\RV{X})}) + \psi_1(\psi_1 + \nfrac{\sigma}{2H(\RV{Y})}) + \psi_2(\psi_2 + \nfrac{\sigma}{2H(\RV{Z})}) +\Exp[f]^2 - 1\rt),\\
    					    \hat{f}^2(\{x,y\})      &= 1 -\Exp[f]^2 -\psi_0^2 - \psi_1^2 - \psi_2(\nfrac{\sigma}{H(\RV{Z})}+\psi_2),\\
    						\hat{f}^2(\{x,z\})      &= 1 -\Exp[f]^2 -\psi_0^2 - \psi_2^2 - \psi_1(\nfrac{\sigma}{H(\RV{Y})}+\psi_1),\\
    						\hat{f}^2(\{y,z\})      &= 1 -\Exp[f]^2 -\psi_1^2 - \psi_2^2 - \psi_0(\nfrac{\sigma}{H(\RV{X})}+\psi_0),\\
    						\hat{f}(\{x\})          &= \psi_0,\\
    						\hat{f}(\{y\})          &= \psi_1,\\
    						\hat{f}(\{z\})          &= \psi_2,
    					\end{split}
				    \end{equation*}
				    where
			        \begin{equation*}
			            \begin{split}
			                \psi_0  &= \nfrac{a_0\sigma}{H(\RV{X})}\bigl(\CI(\RV{T};\RV{X}:\RV{Y}:\RV{Z}) + \CI(\RV{T};\RV{X}:\RV{Y}) + \CI(\RV{T};\RV{X}:\RV{Z}) + \CI(\RV{T};\RV{X}:\RV{Y},\RV{X}:\RV{Z})\\                  
			                        &+\UI(\RV{T};\RV{X}\backslash\RV{Y},\RV{Z})\bigr),\\
			                \psi_1  &= \nfrac{a_1\sigma}{H(\RV{Y})}\bigl(\CI(\RV{T};\RV{X}:\RV{Y}:\RV{Z}) + \CI(\RV{T};\RV{X}:\RV{Y}) + \CI(\RV{T};\RV{Y}:\RV{Z}) + \CI(\RV{T};\RV{X}:\RV{Y},\RV{Y}:\RV{Z})\\
    						        &+\UI(\RV{T};\RV{Y}\backslash\RV{X},\RV{Z})\bigr),\\
    						\psi_2  &= \nfrac{a_2\sigma}{H(\RV{Z})}\bigl(\CI(\RV{T};\RV{X}:\RV{Y}:\RV{Z}) + \CI(\RV{T};\RV{X}:\RV{Z}) + \CI(\RV{T};\RV{Y}:\RV{Z}) + \CI(\RV{T};\RV{X}:\RV{Z},\RV{Y}:\RV{Z})\\
    						        &+ \UI(\RV{T};\RV{Z}\backslash\RV{X},\RV{Y})\bigr).
			            \end{split}
			        \end{equation*}
                \end{corollary}
		% % % % % % % %
		% % % % % % % %
		% %
		\section{Discussion and Future directions}
		    Partial information decomposition aims to quantify an important description of any complex system \cite{sootla2017analyzing}, namely how the information about a part of the system (target) is distributed among several other parts (sources). In particular, PID aims to quantify the synergistic, unique, and redundant information contributions of set of sources about a target. Despite the conceptual importance of such a question and the applications to several fields~\cite{wibral2017partial,frey2018synergistic,ver2017disentangled, kay2017partial, wibral2017modification}, the interpretation of individual PID terms or even of specific PID implementations is still open.   
		   
		   One approach to deepen the understanding of PID terms and extract additional \emph{nontrivial} properties is by drawing a connection between the PID and some related well studied frameworks. For instance, Gutknecht et al.~\cite{gutknecht2020bits} recently formulated the PID problem in the framework of logic and mereology (the study of parthood relationship) deriving the PID terms from parthood relationships between the information contributions of sources and independetly from logical statements about the sources yielding insights into the possibility of quantifying PID terms based on concepts other than redundancy. Another suitable candidate to link PID to is the framework of \emph{Fourier analysis}. The reason is that a specific type of mechanisms, namely, Boolean gates have been exhaustively studied in the scope of their Fourier expansion~\cite{o2014analysis}. This analysis of Boolean mechanisms is not only mathematically rigorous but has lead to the characterization of the relevant properties that governs the behaviour of these mechanisms~\cite{jukna2012boolean,friedgut1998boolean,heckel2013harmonic}. Thus, studying the relation between those Fourier-based properties of Boolean mechanisms and the PID of the information they process sounds promising in revealing insights into understanding PID. 
		     
		    In this paper, we constructed a mapping from the PID terms to the Fourier coefficients of Boolean gates. We work out this map explicitly for the bivariate (linear map) and trivariate (non-linear map) case and obtained its reduction to specific families of these gates such as monotone and unate functions. Using the mapping, we explicit out how the PID terms relate to important properties such as sensitivity, stability and noise sensitivity of Boolean gates that governs their behaviour. We showed that synergy plays an adversarial role in the robustness of the mechanism when its inputs are being perturbed. Finally, we extended the mapping to the $p$-biased case of the bivariate and trivariate gates.
		    
		    Given the complexity of PIDs for multivariate functions (18 terms for trivariate functions and the number grows super exponentially with the number of sources \cite{Williams10}), grouping the numerous terms that appear in PID according to their role in sensitivity and robustness measures may help to identify combinations of terms to which associate intuitive roles. Knowing how different sources of synergistic, redundant, and unique information contribute or cancel each other in sensitivity and robustness measures paves the way for optimizing these measures (subject to other constraints) when learning goals. %or controls can be imposed over different terms of PID. 
		    
		    In summary, we introduced in this paper an approach of studying PID of Boolean functions from the perspective of their Fourier coefficients. A next step should be to convert further results established in the Fourier analysis of Boolean mechanisms in terms of PID. In particular, a future direction aims to provide an interpretation of results such as Friedgut's sharp threshold theorem\cite{friedgut1998boolean} and Russo–Margulis formula~\cite{margulis1974probabilistic,russo1982approximate} about the threshold behavior of monotone functions from the perspective of multivariate PID.  
		    
    \section*{Acknowledgments}
        This research was supported by the Estonian Research Council, ETAG, through PUT Exploratory Grant \#620. A.M.~is employed at the Campus Institute for Dynamics of Biological Networks (CIDBN) funded by the Volkswagen Stiftung. R.V. also thanks the financial support from ETAG through the personal research grant PUT1476. We are also gratefully acknowledge funding by the European Regional Development Fund through the Estonian Center of Excellence in IT, EXCITE.

	\begin{appendix}
	    \section[]{$p$-biased Fourier Analysis}\label{apx:p-biased-bg}
			A Boolean function $f$ has a unique Fourier transform as a multilinear polynomial. To simplify the notation of the Fourier transform, the basis of Boolean function will be taken in $\{-1,1\}$ and so $f:\{-1,1\}^n\to\{-1,1\}$. The Fourier transform of $f$ is given as 
			\begin{equation}
				f(\RV{X}) = \sum_{S\subseteq[n]}\hat{f}(S)\Phi_S(\RV{X}),\quad 
				\Phi_S(\RV{X}) := 
				\begin{cases}
					\prod_{i\in S}\frac{\RV{X}_i -\mu_i}{\sigma_i} &\text{if $S\subseteq[n]\backslash\emptyset$},\\
					1 &\text{otherwise}.
				\end{cases}
			\end{equation}
			The Parseval's identity implies that the Fourier coefficients of $f$ define a probability distribution over $S\subseteq [n],$ 
			\begin{equation}\label{eq:parseval-p}
				\Expl_\RV{X}[f(\RV{X})^2]= \sum_{S\subseteq[n]} \hat{f}(S)^2 = 1.
			\end{equation}
			So, Parseval's identity defines the variance of $f$ in terms of the Fourier coefficients
			\begin{equation}
				\Var[f] = \Expl_\RV{X}[f(\RV{X})^2] - \Expl_\RV{X}[f(\RV{X})]^2 = \sum_{S\subseteq[n]} \hat{f}(S)^2 - \hat{f}(\emptyset) = \sum_{S\subseteq[n]\backslash\emptyset} \hat{f}(S)^2.
			\end{equation}
			Since $\RV{X}_i$ takes the value $-1$ or $1$, the formulas $\phi(1) = \sqrt{\nfrac{p}{1-p}}$ and $\phi(-1) = -\sqrt{\nfrac{1 - p}{p}}$ for each $i\in[n]$ are highlighted. 
			The influence $\Inf_i[f]$ of a source $\RV{X}_i$ over the value of the target $f(x_1,\dots,x_n)$ is defined as the probability of the event that flipping the $i$th source flips the value of the function (target). 
			
			The influence can be determined in terms of the Fourier coefficients~\cite[Proposition 8.45]{o2014analysis} by
			\begin{equation}
				\Inf_i[f]= \frac{1}{\sigma_i^2}\sum_{S\subseteq[n]: i\in S} \hat{f}(S)^2.
			\end{equation}

%%%
			In addition, the influence of a group $A$ of sources is
			\begin{equation}
				\Inf_A[f]= \sum_{S\subseteq[n]} \hat{f}(S)^2\sum_{i\in S\cap A}\frac{1}{\sigma_i^2}.
			\end{equation}
%%%
            Finally, we will state the influences $\Inf_i$ for some families of Boolean functions. 
			\begin{proposition}[Proposition 8.45~\cite{o2014analysis}]\label{prop:mono-inf-p}
				Let $f : \{-1,1\}^n \to\{-1,1\}$ be monotone $p$-biased function. Then, 
				\begin{equation*}
				    \Inf_i[f] = \begin{cases}
				                    \frac{1}{\sigma}\hat{f}(\{i\}) &~\text{if $f$ is increasing}\\
				                    -\frac{1}{\sigma}\hat{f}(\{i\}) &~\text{if $f$ is decreasing.}
				                \end{cases}
				\end{equation*}
			\end{proposition}
			
			\begin{proposition}[Proposition 3~\cite{heckel2013harmonic}]\label{prop:unate-p}
            Let $f:\{-1,1\}^n\to\{-1,1\}$ be $p$-biased unate. Then,
            $$\hat{f}(\{x_i\}) = a_i\sigma\Inf_i[f],$$
            where $a_i\in\{-1,1\}$ is the unate parameter.
            \end{proposition}		

		%%%%%%%%%%%%%%%%%%%%%%%%%
	%%%%%%%%%%%%%%%%%%%%%%%%%%%%%
		%%%%%%%%%%%%%%%%%%%%%%%%%
	    \section{Proof of Theorem~\ref{thm:fou-pid-bi}}\label{apx:sec:proof:thm:fou-pid-bi}
            Since $\RV{X}$ and $\RV{Y}$ are sampled uniformly at random and by~\eqref{eq:inf-mi-uni}, 
            \begin{equation*}
                \begin{split}
                    \Inf_1[f] &= \MI(f(\RV{X},\RV{Y});\RV{X}\mid\RV{Y}),\\
                    \Inf_2[f] &= \MI(f(\RV{X},\RV{Y});\RV{Y}\mid\RV{X}).
                \end{split}				
            \end{equation*}
            Using the basic identities of PID,
            \begin{equation*}
                \begin{split}
                    \MI(f(\RV{X},\RV{Y});\RV{X}\mid\RV{Y}) &= \CI(\RV{T};\RV{X}:\RV{Y}) + \UI(\RV{T};\RV{X}\backslash\RV{Y})\\
                    \MI(f(\RV{X},\RV{Y});\RV{Y}\mid\RV{X}) &= \CI(\RV{T};\RV{X}:\RV{Y}) + \UI(\RV{T};\RV{Y}\backslash\RV{X}),
                \end{split}
            \end{equation*}
            then the relation between PID and Fourier coefficients can be expressed as follows
            \begin{equation*}
                \bMtx{ 1 & 1 & 0\\ 1 & 0 & 1 }\bMtx{ \CI(\RV{T};\RV{X}:\RV{Y})\\ \UI(\RV{T};\RV{X}\backslash\RV{Y})\\ \UI(\RV{T};\RV{Y}\backslash\RV{X}) } 
                = \bMtx{ 1 & 1 & 0\\ 1 & 0 & 1 }\bMtx{\hat{f}(\{X,Y\})^2 \\ \hat{f}(\{X\})^2\\ \hat{f}(\{Y\})^2 }.
            \end{equation*}
            Let $A = \bMtx{ 1 & 1 & 0\\ 1 & 0 & 1 }$, $d= \bMtx{ \CI(\RV{T};\RV{X}:\RV{Y})\\ \UI(\RV{T};\RV{X}\backslash\RV{Y})\\ \UI(\RV{T};\RV{Y}\backslash\RV{X}) }$, and $r= \bMtx{\hat{f}(\{X,Y\})^2 \\ \hat{f}(\{X\})^2\\ \hat{f}(\{Y\})^2 }$,  then the system we are aiming to solve is the following  
            \[
                A r = A d
            \]
            Now, using the Moore-Penrose inverse:
            \begin{equation}\label{eq:fou-pid-bi}
                r = A^+ Ad  + (I_3 - A^+A)w_f,
            \end{equation}
            where $w_f\in\RR^3$. Since $(I_3 - A^+A)$ is an orthogonal projection onto the kernel of $A$, then $(I_3 - A^+A)^2 = (I_3 - A^+A)$ and $(I_3 - A^+A)^+ = (I_3 - A^+A)$. Hence, 
            \[
                \begin{split}
                    r &= A^+ Ad  + (I_3 - A^+A)w_f\\
                    (I_3 - A^+A)r &= (I_3 - A^+A)A^+ Ad  + (I_3 - A^+A)^2w_f\\
                    (I_3 - A^+A)r &= (I_3 - A^+A)w_f\\
                    (I_3 - A^+A)r &= w_f.\\
                \end{split}
            \]
            Moreover, 
            \begin{equation}\label{eq:fou-pid-par-bi}
                \begin{split}
                    w_f &= (I_3 - A^+A)r\\
                        &= \frac{1}{3}(\Stab_{-1}[f] - \Exp[f]^2) u,
                \end{split}
            \end{equation}
            where  $u^T = (1,-1,-1)$.
            Using~\eqref{eq:fou-pid-bi} and~\eqref{eq:fou-pid-par-bi}, the Fourier coefficients can be expressed in terms of $\CI(\RV{T};\RV{X}:\RV{Y}), \UI(\RV{T};\RV{X}\backslash\RV{Y}), \UI(\RV{T};\RV{Y}\backslash\RV{X}), \Stab_{-1}[f],$ and $\Exp[f]^2$ as
            \begin{equation}\label{eq:fou-pid-bi-stab-exp}
                \begin{split}
                    \hat{f}(\{X,Y\})^2 	&= \frac{2}{3} \CI(\RV{T};\RV{X}:\RV{Y}) + \frac{1}{3}\UI(\RV{T};\RV{X}\backslash\RV{Y}) + \frac{1}{3} \UI(\RV{T};\RV{Y}\backslash\RV{X}) + \frac{1}{3}(\Stab_{-1}[f] - \Exp[f]^2),\\ 
                    \hat{f}(\{X\})^2	&= \frac{1}{3} \CI(\RV{T};\RV{X}:\RV{Y}) + \frac{2}{3}\UI(\RV{T};\RV{X}\backslash\RV{Y}) - \frac{1}{3} \UI(\RV{T};\RV{Y}\backslash\RV{X}) - \frac{1}{3}(\Stab_{-1}[f] - \Exp[f]^2),\\ 
                    \hat{f}(\{Y\})^2 	&= \frac{1}{3} \CI(\RV{T};\RV{X}:\RV{Y}) - \frac{1}{3}\UI(\RV{T};\RV{X}\backslash\RV{Y}) + \frac{2}{3} \UI(\RV{T};\RV{Y}\backslash\RV{X}) - \frac{1}{3}(\Stab_{-1}[f] - \Exp[f]^2).
                \end{split}
            \end{equation}
            Using the Parseval identity~\eqref{eq:parseval},  
            \begin{equation}\label{eq:stab-exp-pid-bi}
                \begin{split}
                    \frac{1}{3}\Stab_{-1}[f] 	&=\frac{2}{3} \lt(2\CI(\RV{T};\RV{X}:\RV{Y}) + \UI(\RV{T};\RV{X}\backslash\RV{Y}) +  \UI(\RV{T};\RV{Y}\backslash\RV{X}) + 2\Exp[f]^2\rt) - 1,\\
                    \frac{1}{3}\Exp[f]^2		&=-\frac{1}{6}\lt( 2\CI(\RV{T};\RV{X}:\RV{Y}) + \UI(\RV{T};\RV{X}\backslash\RV{Y}) + \UI(\RV{T};\RV{Y}\backslash\RV{X}) -\frac{1}{2}\Stab_{-1}[f] \rt) + \frac{1}{4}.
                \end{split}
            \end{equation}
Using~\eqref{eq:stab-exp-pid-bi} and~\eqref{eq:exp-mi}, the Fourier coefficients are expressed in terms of $\CI(\RV{T};\RV{X}:\RV{Y}), \UI(\RV{T};\RV{X}\backslash\RV{Y}),$  $\UI(\RV{T};\RV{Y}\backslash\RV{X}),$ and $\Exp(f)^2$ as
            \begin{equation}%\label{eq:fou-pid-bi-exp}
                \begin{split}
                    \hat{f}(\{X,Y\})^2 	&= 2\CI(\RV{T};\RV{X}:\RV{Y}) + \UI(\RV{T};\RV{X}\backslash\RV{Y}) + \UI(\RV{T};\RV{Y}\backslash\RV{X}) + \Exp[f]^2 - 1,\\ 
                    \hat{f}(\{X\})^2	&= 1 - \CI(\RV{T};\RV{X}:\RV{Y}) - \UI(\RV{T};\RV{Y}\backslash\RV{X}) - \Exp[f]^2,\\ 
                    \hat{f}(\{Y\})^2 	&= 1 - \CI(\RV{T};\RV{X}:\RV{Y}) - \UI(\RV{T};\RV{X}\backslash\RV{Y}) - \Exp[f]^2\\
                    h(\frac{1}{2}(1 + \Exp[f])) &= \CI(\RV{T};\RV{X}:\RV{Y}) + \UI(\RV{T};\RV{X}\backslash\RV{Y}) + \UI(\RV{T};\RV{Y}\backslash\RV{X}) + \SI(\RV{T};\RV{X},\RV{Y}),
                \end{split}				
            \end{equation}
            and in terms of $\CI(\RV{T};\RV{X}:\RV{Y}), \UI(\RV{T};\RV{X}\backslash\RV{Y}), \UI(\RV{T};\RV{Y}\backslash\RV{X}),$ and  $\Stab_{-1}[f]$ as 
            \begin{equation}\label{eq:fou-pid-bi-stab}
                \begin{split}
                    \hat{f}(\{X,Y\})^2 	&= \CI(\RV{T};\RV{X}:\RV{Y}) + \frac{1}{2}\UI(\RV{T};\RV{X}\backslash\RV{Y}) + \frac{1}{2}\UI(\RV{T};\RV{Y}\backslash\RV{X}) - \frac{1}{4}(1 - \Stab_{-1}[f]),\\ 
                    \hat{f}(\{X\})^2	&= \frac{1}{2}\UI(\RV{T};\RV{X}\backslash\RV{Y}) - \frac{1}{2}\UI(\RV{T};\RV{Y}\backslash\RV{X}) + \frac{1}{4}(1 - \Stab_{-1}[f]),\\ 
                    \hat{f}(\{Y\})^2 	&= -\frac{1}{2}\UI(\RV{T};\RV{X}\backslash\RV{Y}) + \frac{1}{2}\UI(\RV{T};\RV{Y}\backslash\RV{X}) + \frac{1}{4}( 1 - \Stab_{-1}[f]).
                \end{split}				
            \end{equation}
%%%%%%%%%%%%%%%%%%%%%%%%%%%%
%%%%%%%%%%%%%%%%%%%%%%%%%%%%
%%%%%%%%%%%%%%%%%%%%%%%%%%%%
        \section{Proof of Theorem~\ref{prop:fou-pid-tri}}\label{apx:sec:proof:prop:fou-pid-tri}
            Since $\RV{X},\RV{Y}$ and $\RV{Z}$ are sampled uniformly at random, then by~\eqref{eq:inf-mi-uni},
            \begin{equation}
                \begin{split}
                    \Inf_1(f) &= \MI(f(\RV{X},\RV{Y},\RV{Z});\RV{X}\mid\RV{Y},\RV{Z}),\\
                    \Inf_2(f) &= \MI(f(\RV{X},\RV{Y},\RV{Z});\RV{Y}\mid\RV{X},\RV{Z}),\\
                    \Inf_3(f) &= \MI(f(\RV{X},\RV{Y},\RV{Z});\RV{Z}\mid\RV{X},\RV{Y}).\\
                \end{split}
            \end{equation}
            Using the basic identities of PID,
            \begin{equation*}
                \begin{split}
                    \MI(\RV{T};\RV{X}\mid\RV{Y},\RV{Z}) &= \CI(\RV{T};\RV{X}:\RV{Y}:\RV{Z}) +\CI(\RV{T};\RV{X}:\RV{Y}) + \CI(\RV{T};\RV{X}:\RV{Z}) +  \CI(\RV{T};\RV{X}:\RV{Y},\RV{X}:\RV{Z})\\ 
                                                        &+ \UI(\RV{T};\RV{X}\backslash\RV{Y},\RV{Z}),\\
                    \MI(\RV{T};\RV{Y}\mid\RV{X},\RV{Z}) &= \CI(\RV{T};\RV{X}:\RV{Y}:\RV{Z}) +\CI(\RV{T};\RV{X}:\RV{Y}) + \CI(\RV{T};\RV{Y}:\RV{Z}) +  \CI(\RV{T};\RV{X}:\RV{Y},\RV{Y}:\RV{Z})\\
                                                        &+ \UI(\RV{T};\RV{Y}\backslash\RV{X},\RV{Z}),\\
                    \MI(\RV{T};\RV{Z}\mid\RV{X},\RV{Y}) &= \CI(\RV{T};\RV{X}:\RV{Y}:\RV{Z}) +\CI(\RV{T};\RV{X}:\RV{Z}) + \CI(\RV{T};\RV{Y}:\RV{Z}) +  \CI(\RV{T};\RV{X}:\RV{Z},\RV{Y}:\RV{Z})\\
                                                        &+ \UI(\RV{T};\RV{Z}\backslash\RV{X},\RV{Y}),
                \end{split}
            \end{equation*}
            then the relation between PID and Fourier coefficients is
            \begin{equation}
                A_d d = A_r r 
            \end{equation}
            where 
            \begin{equation*}
                \begin{split}
                    A_d &= \bMtx{	1 & 1 & 1 & 0 & 1 & 0 & 0 & 1 & 0 & 0\\
                                    1 & 1 & 0 & 1 & 0 & 1 & 0 & 0 & 1 & 0\\
                                    1 & 0 & 1 & 1 & 0 & 0 & 1 & 0 & 0 & 1
                                },\\				
                    A_r	&= \bMtx{	1 & 1 & 1 & 0 & 1 & 0 & 0\\ 
                                    1 & 1 & 0 & 1 & 0 & 1 & 0\\
                                    1 & 0 & 1 & 1 & 0 & 0 & 1
                                } 
                \end{split}
            \end{equation*}
            \begin{equation*}
                \begin{split}
                    d^T =	&\bigl(\CI(\RV{T};\RV{X}:\RV{Y}:\RV{Z}), \CI(\RV{T};\RV{X}:\RV{Y}), \CI(\RV{T};\RV{X}:\RV{Z}), \CI(\RV{T};\RV{Y}:\RV{Z}),\CI(\RV{T};\RV{X}:\RV{Y},\RV{X}:\RV{Z}),\\ 		
                            &\CI(\RV{T};\RV{X}:\RV{Y},\RV{Y}:\RV{Z}), \CI(\RV{T};\RV{X}:\RV{Z},\RV{Y}:\RV{Z}), \UI(\RV{T};\RV{X}\backslash\RV{Y},\RV{Z}),\\  
                            & \UI(\RV{T};\RV{Y}\backslash\RV{X},\RV{Z}), \UI(\RV{T};\RV{Z}\backslash\RV{X},\RV{Y}) \bigr),\\
                    r^T = &(\hat{f}(\{X,Y,Z\})^2, \hat{f}(\{X,Y\})^2, \hat{f}(\{X,Z\})^2, \hat{f}(\{Y,Z\})^2, \hat{f}(\{X\})^2, \hat{f}(\{Y\})^2, \hat{f}(\{Z\})^2).
                \end{split}
            \end{equation*}
            Using Moore-Penrose inverse 
            \begin{equation}\label{eq:fou-pid-tri}
                r = A_r^+ A_d d  + (I_7 - A_r^+A_r)w_f
            \end{equation}
            where $w_f\in\RR^7$. Since $I_7 - A_r^+A_r$ is an orthogonal projection onto the Kernel of $A_r$ then $(I_7 - A_r^+A_r)w_f = (I_7 - A_r^+A_r)r$. Unlike the Bivariate case -- see equation~\eqref{eq:fou-pid-par-bi} -- $(I_7 - A_r^+A_r)r$ does not have a nice format, but each entry can be trivially bounded
            \begin{equation*}
                -\frac{2}{8}\le U_0,U_4,U_5,U_6\le \frac{5}{8} \qquad -\frac{2}{8}\le U_1,U_2,U_3\le \frac{4}{8}
            \end{equation*}
            where $U = (I_7 - A_r^+A_r)r$. Hence the squared Fourier coefficients can be bounded accordingly.
%%%%%%%%%%%%%%%%%%%%%%%%
%%%%%%%%%%%%%%%%%%%%%%%%
%%%%%%%%%%%%%%%%%%%%%%%%
        \section{Proof of Theorem~\ref{thm:fou-pid-tri-mono}}\label{apx:sec:proof:thm:fou-pid-tri-mono}
            Since $f$ is a monotone Boolean function, then by Proposition~\ref{prop:mono-inf} and the influence-information relation in~\eqref{eq:inf-mi-uni}, 
            \begin{equation}\label{eq:mono-inf}
				\begin{split}
                    \hat{f}(\{x\})  &= \CI(\RV{T};\RV{X}:\RV{Y}:\RV{Z}) + \CI(\RV{T};\RV{X}:\RV{Y}) + \CI(\RV{T};\RV{X}:\RV{Z}) + \CI(\RV{T};\RV{X}:\RV{Y},\RV{X}:\RV{Z})\\
                                    &+\UI(\RV{T};\RV{X}\backslash\RV{Y},\RV{Z}),\\
                    \hat{f}(\{y\})  &= \CI(\RV{T};\RV{X}:\RV{Y}:\RV{Z}) + \CI(\RV{T};\RV{X}:\RV{Y}) + \CI(\RV{T};\RV{Y}:\RV{Z}) + \CI(\RV{T};\RV{X}:\RV{Y},\RV{Y}:\RV{Z})\\
                                    &+\UI(\RV{T};\RV{Y}\backslash\RV{X},\RV{Z}),\\
                    \hat{f}(\{z\})  &= \CI(\RV{T};\RV{X}:\RV{Y}:\RV{Z}) + \CI(\RV{T};\RV{X}:\RV{Z}) + \CI(\RV{T};\RV{Y}:\RV{Z}) + \CI(\RV{T};\RV{X}:\RV{Z},\RV{Y}:\RV{Z})\\
                                    &+ \UI(\RV{T};\RV{Z}\backslash\RV{X},\RV{Y}).
				\end{split}
		    \end{equation}
            Now using the definition of influences and equation~\eqref{eq:mono-inf}, the following relation between PID and the remaining Fourier coefficients is formulated 
            \begin{equation}
                A_d d = A_r r 
            \end{equation}
            where 
            \begin{equation*}
                \begin{split}
                    A_d &= \bMtx{	1 & 1 & 1 & 0 & 1 & 0 & 0 & 1 & 0 & 0\\
                                    1 & 1 & 0 & 1 & 0 & 1 & 0 & 0 & 1 & 0\\
                                    1 & 0 & 1 & 1 & 0 & 0 & 1 & 0 & 0 & 1
                                },\\				
                    A_r	&= \bMtx{	1 & 1 & 1 & 0 \\ 
                                    1 & 1 & 0 & 1 \\
                                    1 & 0 & 1 & 1 
                                } 
                \end{split}
            \end{equation*}
            \begin{equation*}
                \begin{split}
                    d^T =	&\bigl(\CI(\RV{T};\RV{X}:\RV{Y}:\RV{Z}), \CI(\RV{T};\RV{X}:\RV{Y}), \CI(\RV{T};\RV{X}:\RV{Z}), \CI(\RV{T};\RV{Y}:\RV{Z}),\CI(\RV{T};\RV{X}:\RV{Y},\RV{X}:\RV{Z}),\\ 		
                            &\CI(\RV{T};\RV{X}:\RV{Y},\RV{Y}:\RV{Z}), \CI(\RV{T};\RV{X}:\RV{Z},\RV{Y}:\RV{Z}), \UI(\RV{T};\RV{X}\backslash\RV{Y},\RV{Z}),\\  
                            & \UI(\RV{T};\RV{Y}\backslash\RV{X},\RV{Z}), \UI(\RV{T};\RV{Z}\backslash\RV{X},\RV{Y}) \bigr),\\
                    r^T =   &(\hat{f}(\{X,Y,Z\})^2, \hat{f}(\{X,Y\})^2, \hat{f}(\{X,Z\})^2, \hat{f}(\{Y,Z\})^2).
                \end{split}
            \end{equation*}
            Using Moore-Penrose inverse 
            \begin{equation}\label{eq:fou-pid-tri-mono}
                r = A_r^+ A_d d  + (I_4 - A_r^+A_r)w_f
            \end{equation}
            where $w_f\in\RR^4$. Since $I_4 - A_r^+A_r$ is an orthogonal projection onto the Kernel of $A_r$, then
            \begin{equation}\label{eq:fou-pid-par-tri}
                \begin{split}
                    w_f &= (I_4 - A_r^+A_r)r\\
                        &= \frac{1}{7}(\Stab_{-1}[f] - \hat{f}(\{X,Y,Z\})^2 + \hat{f}(\{X\})^2 + \hat{f}(\{Y\})^2 + \hat{f}(\{Z\})^2 - \Exp[f]^2) u,
                \end{split}
            \end{equation}
            where $u = (-2, 1, 1, 1)$. So using~\eqref{eq:fou-pid-par-tri} and~\eqref{eq:fou-pid-tri-mono}, the Fourier coefficients can be expressed as
            \begin{equation*}
                \begin{split}
                    \hat{f}^2(\{x,y,z\})    &= \nfrac{1}{5}\psi_0 + \nfrac{2}{5}\lt(\Exp[f]^2 - \Stab_{-1}[f] - \psi_4^2 - \psi_5^2 - \psi_6^2\rt),\\
                    \hat{f}^2(\{x,y\})      &= \nfrac{1}{7}\lt(\psi_1 - \nfrac{1}{5}\psi_0 \rt) + \nfrac{1}{5}\lt(-\Exp[f]^2 + \Stab_{-1}[f] + \psi_4^2 + \psi_5^2 + \psi_6^2\rt),\\
                    \hat{f}^2(\{x,z\})      &= \nfrac{1}{7}\lt(\psi_2 - \nfrac{1}{5}\psi_0 \rt) + \nfrac{1}{5}\lt(-\Exp[f]^2 + \Stab_{-1}[f] + \psi_4^2 + \psi_5^2 + \psi_6^2\rt),\\
                    \hat{f}^2(\{y,z\})      &= \nfrac{1}{7}\lt(\psi_3 - \nfrac{1}{5}\psi_0 \rt) + \nfrac{1}{5}\lt(-\Exp[f]^2 + \Stab_{-1}[f] + \psi_4^2 + \psi_5^2 + \psi_6^2\rt),\\
                    \hat{f}(\{x\})          &= \psi_4,\\
                    \hat{f}(\{y\})          &= \psi_5,\\
                    \hat{f}(\{z\})          &= \psi_6,\\
                \end{split}
            \end{equation*}
            where
            \begin{equation*}
                \begin{split}
                    \psi_0  &= 3\CI(\RV{T};\RV{X}:\RV{Y}:\RV{Z}) + 2\CI(\RV{T};\RV{X}:\RV{Y}) + 2\CI(\RV{T};\RV{X}:\RV{Z}) + 2\CI(\RV{T};\RV{Y}:\RV{Z})\\
                            &+ \CI(\RV{T};\RV{X}:\RV{Y},\RV{X}:\RV{Z}) + \CI(\RV{T};\RV{X}:\RV{Y},\RV{X}:\RV{Z}) + \CI(\RV{T};\RV{X}:\RV{Y},\RV{X}:\RV{Z})\\
                            &+\UI(\RV{T};\RV{X}\backslash\RV{Y},\RV{Z}) + \UI(\RV{T};\RV{Y}\backslash\RV{X},\RV{Z}) + \UI(\RV{T};\RV{Z}\backslash\RV{X},\RV{Y}),\\
                    \psi_1  &= 2\CI(\RV{T};\RV{X}:\RV{Y}:\RV{Z}) + 6\CI(\RV{T};\RV{X}:\RV{Y}) - \CI(\RV{T};\RV{X}:\RV{Z}) - \CI(\RV{T};\RV{Y}:\RV{Z})\\ 
                            &+ 3\CI(\RV{T};\RV{X}:\RV{Y},\RV{X}:\RV{Z}) + 3\CI(\RV{T};\RV{X}:\RV{Y},\RV{X}:\RV{Z}) - 4\CI(\RV{T};\RV{X}:\RV{Y},\RV{X}:\RV{Z})\\
                            &+ 3\UI(\RV{T};\RV{X}\backslash\RV{Y},\RV{Z}) + 3\UI(\RV{T};\RV{Y}\backslash\RV{X},\RV{Z}) - 4\UI(\RV{T};\RV{Z}\backslash\RV{X},\RV{Y}),\\
                    \psi_2  &= 2\CI(\RV{T};\RV{X}:\RV{Y}:\RV{Z}) - \CI(\RV{T};\RV{X}:\RV{Y}) + 6\CI(\RV{T};\RV{X}:\RV{Z}) - \CI(\RV{T};\RV{Y}:\RV{Z})\\
                            &+ 3\CI(\RV{T};\RV{X}:\RV{Y},\RV{X}:\RV{Z}) - 4\CI(\RV{T};\RV{X}:\RV{Y},\RV{X}:\RV{Z}) + 3\CI(\RV{T};\RV{X}:\RV{Y},\RV{X}:\RV{Z})\\
                            &+ 3\UI(\RV{T};\RV{X}\backslash\RV{Y},\RV{Z}) - 4\UI(\RV{T};\RV{Y}\backslash\RV{X},\RV{Z}) + 3\UI(\RV{T};\RV{Z}\backslash\RV{X},\RV{Y}),\\
                    \psi_3  &= 2\CI(\RV{T};\RV{X}:\RV{Y}:\RV{Z}) - \CI(\RV{T};\RV{X}:\RV{Y}) - \CI(\RV{T};\RV{X}:\RV{Z}) + 6\CI(\RV{T};\RV{Y}:\RV{Z})\\
                            &- 4\CI(\RV{T};\RV{X}:\RV{Y},\RV{X}:\RV{Z}) + 3\CI(\RV{T};\RV{X}:\RV{Y},\RV{X}:\RV{Z}) + 3\CI(\RV{T};\RV{X}:\RV{Y},\RV{X}:\RV{Z})\\     
                            &- 4\UI(\RV{T};\RV{X}\backslash\RV{Y},\RV{Z}) + 3\UI(\RV{T};\RV{Y}\backslash\RV{X},\RV{Z}) + 3\UI(\RV{T};\RV{Z}\backslash\RV{X},\RV{Y}),\\
                    \psi_4  &= \CI(\RV{T};\RV{X}:\RV{Y}:\RV{Z}) + \CI(\RV{T};\RV{X}:\RV{Y}) + \CI(\RV{T};\RV{X}:\RV{Z}) + \CI(\RV{T};\RV{X}:\RV{Y},\RV{X}:\RV{Z})\\
                            &+\UI(\RV{T};\RV{X}\backslash\RV{Y},\RV{Z}),\\
                    \psi_5  &= \CI(\RV{T};\RV{X}:\RV{Y}:\RV{Z}) + \CI(\RV{T};\RV{X}:\RV{Y}) + \CI(\RV{T};\RV{Y}:\RV{Z}) + \CI(\RV{T};\RV{X}:\RV{Y},\RV{Y}:\RV{Z})\\
                            &+\UI(\RV{T};\RV{Y}\backslash\RV{X},\RV{Z}),\\
                    \psi_6  &= \CI(\RV{T};\RV{X}:\RV{Y}:\RV{Z}) + \CI(\RV{T};\RV{X}:\RV{Z}) + \CI(\RV{T};\RV{Y}:\RV{Z}) + \CI(\RV{T};\RV{X}:\RV{Z},\RV{Y}:\RV{Z})\\
                            &+ \UI(\RV{T};\RV{Z}\backslash\RV{X},\RV{Y}).
                \end{split}
            \end{equation*}
            But Parseval identity~\eqref{eq:parseval} allows to get the following relation between PID terms, $\Stab_{-1}[f],$ and $\Exp[f]$   
            \begin{equation}\label{eq:stab-exp-pid-tri-mono}
                \begin{split}
                    \frac{1}{5}\Stab_{-1}[f] 	&=-\frac{1}{7} \lt(\frac{4}{5}\psi_0 + \psi_1 + \psi_2 + \psi_3 \rt) - \frac{1}{5}\lt(4\Exp[f]^2 + 6\psi_4^2 + 6\psi_5^2 + 6\psi_6^2\rt) + 1.
                \end{split}
            \end{equation}
            Hence, using~\eqref{eq:fou-pid-par-tri}, ~\eqref{eq:fou-pid-tri-mono}, and~\eqref{eq:stab-exp-pid-tri-mono} the theorem is concluded. 
%%%%%%%%%%%%%%%%%
%%%%%%%%%%%%%%%%%
%%%%%%%%%%%%%%%%%
        \section{Proof of Theorem~\ref{thm:fou-pid-bi-p}}\label{apx:sec:proof:thm:fou-pid-bi-p}
            For $p$-biased Boolean functions $H(X_i) = h(p)$ for all $i\in[2]$. Then from~\eqref{eq:inf-mi},
            \begin{equation}
                \begin{split}
                    h(p)\Inf_1[f] &= \MI(f(\RV{X},\RV{Y}),\RV{X}\mid\RV{Y})\\
                    h(p)\Inf_2[f] &= \MI(f(\RV{X},\RV{Y}),\RV{Y}\mid\RV{X}).
                \end{split}
            \end{equation}
            Using the definition of $\Inf_i[f]$ and the relation of Fourier coefficients to PID in the uniformly distributed bivariate case, mutatis mutandis,
            \begin{equation}\label{eq:fou-pid-bi-p-stab-exp}
                \begin{split}
                    \hat{f}(\{X,Y\})^2 	&= \frac{\sigma}{3h(p)}\bigl( 2\CI(\RV{T};\RV{X}:\RV{Y}) + \UI(\RV{T};\RV{X}\backslash\RV{Y}) +  \UI(\RV{T};\RV{Y}\backslash\RV{X})\bigr) + \frac{1}{3}\alpha_f,\\ 
                    \hat{f}(\{X\})^2	&= \frac{\sigma}{3h(p)} \bigl(\CI(\RV{T};\RV{X}:\RV{Y}) + 2\UI(\RV{T};\RV{X}\backslash\RV{Y}) -  \UI(\RV{T};\RV{Y}\backslash\RV{X})\bigr) - \frac{1}{3}\alpha_f,\\ 
                    \hat{f}(\{Y\})^2 	&= \frac{\sigma}{3h(p)}\bigl(\CI(\RV{T};\RV{X}:\RV{Y}) - \UI(\RV{T};\RV{X}\backslash\RV{Y}) + 2\UI(\RV{T};\RV{Y}\backslash\RV{X})\bigr) - \frac{1}{3}\alpha_f,
                \end{split}
            \end{equation}
            where $\alpha = \hat{f}(\{X,Y\})^2 - \hat{f}(\{X\})^2 - \hat{f}(\{Y\})^2.$ Using the Parseval identity~\eqref{eq:parseval-p},
            \begin{equation}\label{eq:stab-exp-pid-bi-p}
                \begin{split}
                    \frac{1}{3}\alpha_f 	&=\frac{2\sigma}{3h(p)} \lt(2\CI(\RV{T};\RV{X}:\RV{Y}) + \UI(\RV{T};\RV{X}\backslash\RV{Y}) +  \UI(\RV{T};\RV{Y}\backslash\RV{X})\rt) +\Exp[f]^2 - 1.
                \end{split}
            \end{equation}
            Using~\eqref{eq:stab-exp-pid-bi-p}, the Fourier coefficients are expressed in terms of $\CI(\RV{T};\RV{X}:\RV{Y}), \UI(\RV{T};\RV{X}\backslash\RV{Y}),$  $\UI(\RV{T};\RV{Y}\backslash\RV{X}),$ and $\Exp(f)^2$ as
            \begin{equation}
                \begin{split}
                    \hat{f}(\{X,Y\})^2 	&= \frac{\sigma}{h(p)}\bigl(2\CI(\RV{T};\RV{X}:\RV{Y}) + \UI(\RV{T};\RV{X}\backslash\RV{Y}) + \UI(\RV{T};\RV{Y}\backslash\RV{X})\bigr) + \Exp[f]^2 - 1,\\ 
                    \hat{f}(\{X\})^2	&= 1 - \frac{\sigma}{h(p)}\bigl(\CI(\RV{T};\RV{X}:\RV{Y}) + \UI(\RV{T};\RV{Y}\backslash\RV{X})\bigr) - \Exp[f]^2,\\ 
                    \hat{f}(\{Y\})^2 	&= 1 - \frac{\sigma}{h(p)}\bigl(\CI(\RV{T};\RV{X}:\RV{Y}) + \UI(\RV{T};\RV{X}\backslash\RV{Y})\bigr) - \Exp[f]^2.\\
                \end{split}				
            \end{equation}
%%%%%%%%%%%%%%%%%%%%
%%%%%%%%%%%%%%%%%%%%
%%%%%%%%%%%%%%%%%%%%
        \section{Mappings of Propositions~\ref{prop:fou-pid-tri} and~\ref{prop:fou-pid-tri-p-biased}}\label{apx:sec:apx-maps}
            The formulas of $\Phi$ the approximate map from $\mathfrak{D}$ to $\mathfrak{F}$ for the trivariate case when sources are uniformly distributed and biased respectively. 
            {\small
            
            \begin{equation}\label{eq:apx-phi-tri-1}
    				\begin{split}
    					\Phi_0(\mathfrak{D})	&=\frac{1}{8}\bigl( 
    						3\CI(\RV{T};\RV{X}:\RV{Y}:\RV{Z}) + 2\CI(\RV{T};\RV{X}:\RV{Y}) + 2\CI(\RV{T};\RV{X}:\RV{Z}) + 2\CI(\RV{T};\RV{Y}:\RV{Z})\\
    						&+ \CI(\RV{T};\RV{X}:\RV{Y},\RV{X}:\RV{Z}) + \CI(\RV{T};\RV{X}:\RV{Y},\RV{Y}:\RV{Z}) + \CI(\RV{T};\RV{X}:\RV{Z},\RV{Y}:\RV{Z})\\
    						&+ \UI(\RV{T};\RV{X}\backslash\RV{Y},\RV{Z}) + \UI(\RV{T};\RV{Y}\backslash\RV{X},\RV{Z}) +\UI(\RV{T};\RV{Z}\backslash\RV{X},\RV{Y})
    						\bigr),\\
    					\Phi_1(\mathfrak{D})	&=\frac{1}{8}\bigl( 
    						2\CI(\RV{T};\RV{X}:\RV{Y}:\RV{Z}) + 4\CI(\RV{T};\RV{X}:\RV{Y}) + 2\CI(\RV{T};\RV{X}:\RV{Y},\RV{X}:\RV{Z}) + 2\CI(\RV{T};\RV{X}:\RV{Y},\RV{Y}:\RV{Z})\\
    						&- 2\CI(\RV{T};\RV{X}:\RV{Z},\RV{Y}:\RV{Z}) + 2\UI(\RV{T};\RV{X}\backslash\RV{Y},\RV{Z}) + 2\UI(\RV{T};\RV{Y}\backslash\RV{X},\RV{Z}) - 2\UI(\RV{T};\RV{Z}\backslash\RV{X},\RV{Y})
    						\bigr),\\
    					\Phi_2(\mathfrak{D})	&=\frac{1}{8}\bigl( 
    						2\CI(\RV{T};\RV{X}:\RV{Y}:\RV{Z}) + 4\CI(\RV{T};\RV{X}:\RV{Z}) + 2\CI(\RV{T};\RV{X}:\RV{Y},\RV{X}:\RV{Z}) - 2\CI(\RV{T};\RV{X}:\RV{Y},\RV{Y}:\RV{Z})\\
    						&+ 2\CI(\RV{T};\RV{X}:\RV{Z},\RV{Y}:\RV{Z}) + 2\UI(\RV{T};\RV{X}\backslash\RV{Y},\RV{Z}) - 2\UI(\RV{T};\RV{Y}\backslash\RV{X},\RV{Z}) + 2\UI(\RV{T};\RV{Z}\backslash\RV{X},\RV{Y})
    						\bigr),\\
    					\Phi_3(\mathfrak{D})	&=\frac{1}{8}\bigl( 
    						2\CI(\RV{T};\RV{X}:\RV{Y}:\RV{Z}) + 4\CI(\RV{T};\RV{Y}:\RV{Z}) - 2\CI(\RV{T};\RV{X}:\RV{Y},\RV{X}:\RV{Z}) + 2\CI(\RV{T};\RV{X}:\RV{Y},\RV{Y}:\RV{Z})\\
    						&+ 2\CI(\RV{T};\RV{X}:\RV{Z},\RV{Y}:\RV{Z}) - 2\UI(\RV{T};\RV{X}\backslash\RV{Y},\RV{Z}) + 2\UI(\RV{T};\RV{Y}\backslash\RV{X},\RV{Z}) + 2\UI(\RV{T};\RV{Z}\backslash\RV{X},\RV{Y})
    						\bigr),
    				\end{split}
			    \end{equation}

			    \begin{equation}\label{eq:apx-phi-tri-2}
    				\begin{split}
    					\Phi_3(\mathfrak{D})	&=\frac{1}{8}\bigl( 
    						2\CI(\RV{T};\RV{X}:\RV{Y}:\RV{Z}) + 4\CI(\RV{T};\RV{Y}:\RV{Z}) - 2\CI(\RV{T};\RV{X}:\RV{Y},\RV{X}:\RV{Z}) + 2\CI(\RV{T};\RV{X}:\RV{Y},\RV{Y}:\RV{Z})\\
    						&+ 2\CI(\RV{T};\RV{X}:\RV{Z},\RV{Y}:\RV{Z}) - 2\UI(\RV{T};\RV{X}\backslash\RV{Y},\RV{Z}) + 2\UI(\RV{T};\RV{Y}\backslash\RV{X},\RV{Z}) + 2\UI(\RV{T};\RV{Z}\backslash\RV{X},\RV{Y})
    						\bigr),\\
    					\Phi_4(\mathfrak{D})	&=\frac{1}{8}\bigl( 
    						\CI(\RV{T};\RV{X}:\RV{Y}:\RV{Z}) + 2\CI(\RV{T};\RV{X}:\RV{Y}) + 2\CI(\RV{T};\RV{X}:\RV{Z}) - 2\CI(\RV{T};\RV{Y}:\RV{Z})\\
    						&+ 3\CI(\RV{T};\RV{X}:\RV{Y},\RV{X}:\RV{Z}) - \CI(\RV{T};\RV{X}:\RV{Y},\RV{Y}:\RV{Z}) - \CI(\RV{T};\RV{X}:\RV{Z},\RV{Y}:\RV{Z})\\
    						&+ 3\UI(\RV{T};\RV{X}\backslash\RV{Y},\RV{Z}) - \UI(\RV{T};\RV{Y}\backslash\RV{X},\RV{Z}) - \UI(\RV{T};\RV{Z}\backslash\RV{X},\RV{Y}) \bigr),\\
    					\Phi_5(\mathfrak{D}) &=\frac{1}{8}\bigl( 
    						\CI(\RV{T};\RV{X}:\RV{Y}:\RV{Z}) + 2\CI(\RV{T};\RV{X}:\RV{Y}) - 2\CI(\RV{T};\RV{X}:\RV{Z}) + 2\CI(\RV{T};\RV{Y}:\RV{Z})\\
    						&- \CI(\RV{T};\RV{X}:\RV{Y},\RV{X}:\RV{Z}) + 3\CI(\RV{T};\RV{X}:\RV{Y},\RV{Y}:\RV{Z}) - \CI(\RV{T};\RV{X}:\RV{Z},\RV{Y}:\RV{Z})\\
    						&- \UI(\RV{T};\RV{X}\backslash\RV{Y},\RV{Z}) + 3\UI(\RV{T};\RV{Y}\backslash\RV{X},\RV{Z}) - \UI(\RV{T};\RV{Z}\backslash\RV{X},\RV{Y}) \bigr),\\
    		 			\Phi_6(\mathfrak{D})	&=\frac{1}{8}\bigl( 
    						\CI(\RV{T};\RV{X}:\RV{Y}:\RV{Z}) - 2\CI(\RV{T};\RV{X}:\RV{Y}) + 2\CI(\RV{T};\RV{X}:\RV{Z}) + 2\CI(\RV{T};\RV{Y}:\RV{Z})\\
    						&- \CI(\RV{T};\RV{X}:\RV{Y},\RV{X}:\RV{Z}) - \CI(\RV{T};\RV{X}:\RV{Y},\RV{Y}:\RV{Z}) + 3\CI(\RV{T};\RV{X}:\RV{Z},\RV{Y}:\RV{Z})\\
    						&- \UI(\RV{T};\RV{X}\backslash\RV{Y},\RV{Z}) - \UI(\RV{T};\RV{Y}\backslash\RV{X},\RV{Z}) + 3\UI(\RV{T};\RV{Z}\backslash\RV{X},\RV{Y}) \bigr).\\
    				\end{split}
			    \end{equation}
			    \begin{equation}\label{eq:apx-phi-tri-p-biased-1}
    				\begin{split}
    					\Phi_0(\mathfrak{D})	&=\frac{\sigma}{8h(p)}\bigl( 
    						3\CI(\RV{T};\RV{X}:\RV{Y}:\RV{Z}) + 2\CI(\RV{T};\RV{X}:\RV{Y}) + 2\CI(\RV{T};\RV{X}:\RV{Z}) + 2\CI(\RV{T};\RV{Y}:\RV{Z})\\
    						&+ \CI(\RV{T};\RV{X}:\RV{Y},\RV{X}:\RV{Z}) + \CI(\RV{T};\RV{X}:\RV{Y},\RV{Y}:\RV{Z}) + \CI(\RV{T};\RV{X}:\RV{Z},\RV{Y}:\RV{Z})\\
    						&+ \UI(\RV{T};\RV{X}\backslash\RV{Y},\RV{Z}) + \UI(\RV{T};\RV{Y}\backslash\RV{X},\RV{Z}) +\UI(\RV{T};\RV{Z}\backslash\RV{X},\RV{Y})
    						\bigr),\\
    					\Phi_1(\mathfrak{D})	&=\frac{\sigma}{8h(p)}\bigl( 
    						2\CI(\RV{T};\RV{X}:\RV{Y}:\RV{Z}) + 4\CI(\RV{T};\RV{X}:\RV{Y}) + 2\CI(\RV{T};\RV{X}:\RV{Y},\RV{X}:\RV{Z}) \\
    						&+ 2\CI(\RV{T};\RV{X}:\RV{Y},\RV{Y}:\RV{Z}) - 2\CI(\RV{T};\RV{X}:\RV{Z},\RV{Y}:\RV{Z}) + 2\UI(\RV{T};\RV{X}\backslash\RV{Y},\RV{Z})\\ 
    						&+ 2\UI(\RV{T};\RV{Y}\backslash\RV{X},\RV{Z}) - 2\UI(\RV{T};\RV{Z}\backslash\RV{X},\RV{Y})
    						\bigr),\\
    					\Phi_2(\mathfrak{D})	&=\frac{\sigma}{8h(p)}\bigl( 
    						2\CI(\RV{T};\RV{X}:\RV{Y}:\RV{Z}) + 4\CI(\RV{T};\RV{X}:\RV{Z}) + 2\CI(\RV{T};\RV{X}:\RV{Y},\RV{X}:\RV{Z}) \\
    						&- 2\CI(\RV{T};\RV{X}:\RV{Y},\RV{Y}:\RV{Z}) + 2\CI(\RV{T};\RV{X}:\RV{Z},\RV{Y}:\RV{Z}) + 2\UI(\RV{T};\RV{X}\backslash\RV{Y},\RV{Z})\\ 
    						&- 2\UI(\RV{T};\RV{Y}\backslash\RV{X},\RV{Z}) + 2\UI(\RV{T};\RV{Z}\backslash\RV{X},\RV{Y})
    						\bigr),
    				\end{split}
			    \end{equation}
			    \begin{equation}\label{eq:apx-phi-tri-p-biased-2}
    				\begin{split}
    					\Phi_5(\mathfrak{D}) &=\frac{\sigma}{8h(p)}\bigl( 
    						\CI(\RV{T};\RV{X}:\RV{Y}:\RV{Z}) + 2\CI(\RV{T};\RV{X}:\RV{Y}) - 2\CI(\RV{T};\RV{X}:\RV{Z}) + 2\CI(\RV{T};\RV{Y}:\RV{Z})\\
    						&- \CI(\RV{T};\RV{X}:\RV{Y},\RV{X}:\RV{Z}) + 3\CI(\RV{T};\RV{X}:\RV{Y},\RV{Y}:\RV{Z}) - \CI(\RV{T};\RV{X}:\RV{Z},\RV{Y}:\RV{Z})\\
    						&- \UI(\RV{T};\RV{X}\backslash\RV{Y},\RV{Z}) + 3\UI(\RV{T};\RV{Y}\backslash\RV{X},\RV{Z}) - \UI(\RV{T};\RV{Z}\backslash\RV{X},\RV{Y}) \bigr),\\
    		 			\Phi_6(\mathfrak{D})	&=\frac{\sigma}{8h(p)}\bigl( 
    						\CI(\RV{T};\RV{X}:\RV{Y}:\RV{Z}) - 2\CI(\RV{T};\RV{X}:\RV{Y}) + 2\CI(\RV{T};\RV{X}:\RV{Z}) + 2\CI(\RV{T};\RV{Y}:\RV{Z})\\
    						&- \CI(\RV{T};\RV{X}:\RV{Y},\RV{X}:\RV{Z}) - \CI(\RV{T};\RV{X}:\RV{Y},\RV{Y}:\RV{Z}) + 3\CI(\RV{T};\RV{X}:\RV{Z},\RV{Y}:\RV{Z})\\
    						&- \UI(\RV{T};\RV{X}\backslash\RV{Y},\RV{Z}) - \UI(\RV{T};\RV{Y}\backslash\RV{X},\RV{Z}) + 3\UI(\RV{T};\RV{Z}\backslash\RV{X},\RV{Y}) \bigr).
    				\end{split}
			    \end{equation}
			 
			 }
%%%%%%%%%%%%%%%%%%%%%%%%%%%%
%%%%%%%%%%%%%%%%%%%%%%%%%%%%
%%%%%%%%%%%%%%%%%%%%%%%%%%%%
	\end{appendix}
\end{document}